%% file: main.tex
\documentclass[conference]{IEEEtran}
\IEEEoverridecommandlockouts
\usepackage{cite}
\usepackage{amsmath,amssymb,amsfonts}
\usepackage{amsthm}
\usepackage{graphicx}
\usepackage{textcomp}

\usepackage{booktabs} % For formal tables
\usepackage{multirow}
\usepackage{mathtools}
\usepackage{varwidth}
\usepackage{stmaryrd}

\usepackage{algorithm}
\usepackage[noend]{algpseudocode}
\usepackage{xcolor}
\usepackage{caption}
\input{header}
\newtheorem{definition}{Definition}
\newtheorem{theorem}{Theorem}

\def\BibTeX{{\rm B\kern-.05em{\sc i\kern-.025em b}\kern-.08em
    T\kern-.1667em\lower.7ex\hbox{E}\kern-.125emX}}

\begin{document}

\title{CryptoRec: Privacy-preserving Recommendation as a Service\\
}

\author{\IEEEauthorblockN{ Jun Wang}
\IEEEauthorblockA{\textit{University of Luxembourg} \\
Luxembourg \\
jun.wang@uni.lu}\\
\and
\IEEEauthorblockN{Afonso  Arriaga}
\IEEEauthorblockA{\textit{University of Luxembourg} \\
Luxembourg \\
afonso.delerue@uni.lu}\\
\and
\IEEEauthorblockN{Qiang Tang}
\IEEEauthorblockA{\textit{Luxembourg Institute} \\
\textit{ of Science and Technology} \\
Luxembourg \\
tonyrhul@gmail.com}\\
\and
\IEEEauthorblockN{Peter Y.A. Ryan}
\IEEEauthorblockA{\textit{University of Luxembourg} \\
Luexembourg \\
peter.ryan@uni.lu}
}

\maketitle

\begin{abstract}
Recommender systems rely on large datasets of historical data and entail serious privacy risks. A server offering Recommendation as a Service to a client might leak more information than necessary regarding its recommendation model and dataset. At the same time, the disclosure of the client's preferences to the server is also a matter of concern. Devising privacy-preserving protocols using general cryptographic primitives (e.g., secure multi-party computation or homomorphic encryption), is a typical approach to overcome privacy concerns,  but in conjunction with state-of-the-art recommender systems often yields far-from-practical solutions.

In this paper, we tackle this problem from the direction of constructing crypto-friendly machine learning algorithms. In particular, we propose CryptoRec, a secure framework for Recommendation as a Service, which encompasses a homomorphic encryption friendly recommender system. This recommendation model possesses two interesting properties: (1) It models user-item interactions in an item-only latent feature space in which personalized user representations are automatically captured by aggregating pre-learned item features. This means that a server with a pre-trained model can provide recommendations for a client whose data is not in its training set.  Nevertheless, re-training the model with the client's data still improves accuracy. (2) It only uses addition and multiplication operations,  making the model straightforwardly compatible with homomorphic encryption schemes.

We demonstrate the efficiency and accuracy of CryptoRec on three real-world datasets. CryptoRec allows a server with thousands of items to privately answer a prediction query within a few seconds on a single PC, while its prediction accuracy is still competitive with state-of-the-art recommender systems computing over clear data.
\end{abstract}


\input{introduction}

\input{preliminary}
\input{security}
\input{newmodel}

\input{experiment}

\input{related}

\input{conclusion}
\bibliographystyle{IEEEtran}
\bibliography{ref-lol}

\end{document}

%% file: header.tex
\definecolor{darkblue}{rgb}{0,0,0.5} 
\definecolor{darkred}{rgb}{0.4,0,0} 
\definecolor{lightred}{rgb}{0.8,0,0} 
\definecolor{darkgreen}{rgb}{0,0.5,0} 
\definecolor{lightgray}{gray}{0.90}

\mathchardef\myhyphen ="2D\newcommand{\dash}{\ensuremath{\myhyphen}}

\newcommand{\st}{\mathsf{st}}
\newcommand{\heading}[1]{{\vspace{6pt}\noindent\sc{\textsc{#1.}}}}
\newcommand{\secpar}{\ensuremath{1^\lambda}}

\newcommand{\sample}{\ {\leftarrow}{\mbox{\tiny \$}} \ }
\newcommand{\get}{\leftarrow}
\newcommand{\Adv}{\mathbf{Adv}}

\newcommand{\Client}{\mathrm{Client}}
\newcommand{\Server}{\mathrm{Server}}
\newcommand{\A}{\mathcal{A}}
\newcommand{\Sim}{\mathcal{S}}

\newcommand{\HE}{\mathsf{HE}}
\newcommand{\HESetup}{\mathsf{HE.Setup}}
\newcommand{\HEEnc}{\mathsf{HE.Enc}}
\newcommand{\HEDec}{\mathsf{HE.Dec}}
\newcommand{\HEEval}{\mathsf{HE.Eval}}
\newcommand{\sk}{\mathsf{sk}}
\newcommand{\pk}{\mathsf{pk}}
\newcommand{\msg}{\mathsf{m}}
\newcommand{\ctxt}{\mathsf{c}}
\newcommand{\INDCPA}{\mathrm{IND}\dash\mathrm{CPA}}
\newcommand{\indcpa}{\mathrm{ind}\dash\mathrm{cpa}}

\newcommand{\algrule}[1][.2pt]{\par\vskip.5\baselineskip\hrule height #1\par\vskip.5\baselineskip}

%% file: introduction.tex
\section{Introduction}
Recommender system is one of the most frequently used machine learning applications, which is closely linked to our daily lives. For example, users can always receive personalized item recommendations (e.g., products, videos, etc.) when visiting websites like Amazon and Youtube. This is because their recommender systems are estimating user preferences by analyzing massive historical data, such as browsing records and locations. On one hand, a user can efficiently get preferred products from a vast number of items due to the recommender system; on the other hand, the user data is exposed to the service provider and can be abused~\cite{canny2002collaborative,NS2008}. As such, it results in immediate privacy risks to the user. In this paper, we study how to obtain efficient and accurate recommendation services while preserving data privacy. As an illustrative example, consider the following scenario:

A user (client) with some private data (e.g., ratings) would like to buy a recommendation service to efficiently figure out the most favorable products from a large number of potential candidates. A service provider (e.g., Amazon) has already collected a large database of ratings given by its users on sold items and wishes to monetize its data by selling Recommendation as a Service (RaaS). Different from existing recommender systems, in our scenario the client is unwilling to expose her data to the service provider due to the worries of privacy leakage. At the same time, commercial concerns or requirements may prevent the service provider from releasing its trained recommendation model to the public. In addition, releasing a trained model may also bring privacy risks to the users in the service provider's database. 

We can formalize the above scenario as a secure two-party computation (2PC) protocol. We first describe the Recommendation as a Service (RaaS) as a 2PC protocol, where on one side we have the Server (service provider) with its training data and on the other side the Client with her input. When the protocol terminates, the Client learns a prediction for her input. For obvious reasons, the protocol is only as useful to the Client as the accuracy of the predictions. Then we define the security of the 2PC protocol from two aspects: (1) the Client should only learn the predictions (including what can be inferred from the predictions); (2) the Server should not learn anything about the Client's input. General cryptographic primitives such as secure multi-party computation (SMC) and homomorphic encryption ($\HE$) are immediate candidates to overcome these security concerns.

State-of-the-art recommender systems often rely on non-linear operations, or require training the recommendation model with the Client's data~\cite{zhang2017deep,koren2009matrix,mnih2008probabilistic,zhang2017autosvd}. Generic solutions usually come at a prohibitive cost in terms of efficiency. While improving cryptographic tools (e.g., $\HE$ or SMC) is one typical way to achieve more efficient privacy-preserving solutions, unfortunately, the improvement is usually far from satisfactory to make these solutions practical enough. Recently, CryptoNets~\cite{gilad2016cryptonets} and MiniONN~\cite{liu2017oblivious} have been proposed for privacy-preserving neural networks based Machine Learning as a Service (MLaaS), the scenario of which is similar to ours. The primary contribution of CryptoNets and MiniONN is how to efficiently compute non-linear operations on encrypted data. Instead of using state-of-the-art non-linear activation functions such as ReLu ($relu(x)=max(0,x)$), CryptoNets proposed using a square activation function ($f(x)=x^2$) to avoid non-linear operations, to facilitate evaluating neural networks on encrypted data. This approach may result in a significant accuracy loss \cite{liu2017oblivious}. % non-linear operations can be approximated by addition and multiplication operations, it would immediately result in an efficiency bottleneck ~\cite{Gentry2009,liu2017oblivious}. 
MiniONN introduced a multi-round interactive protocol based on $\HE$ and garbled circuits~\cite{kolesnikov2008improved}, in which non-linear operations were computed by interactions between the Server and the Client. This method requires the two parties to be online constantly, which may increase the difficulty of using MLaaS.

\emph{\textbf{Our contributions.}} We tackle this problem from the direction of designing crypto-friendly machine learning algorithms, so that we can achieve efficient solutions by directly using existing cryptographic tools. In particular, we propose CryptoRec, a new non-interactive secure 2PC protocol for RaaS, the key technical innovation of which is an $\HE$-friendly recommender system. This recommendation model possesses two important properties: (1) It uses only addition and multiplication operations, so that it is straightforwardly compatible with $\HE$ schemes. With this property, CryptoRec is able to complete recommendation computations without requiring the Server and the Client to be online continuously. Simply put, the Client sends her encrypted rating vector to the Server, then the Server computes recommendations with the Client's input and returns the results in an encrypted form. In addition to this, there is no other interaction between the two parties; (2) It can automatically extract personalized user representations by aggregating pre-learned item features, that we say the model has an item-only latent feature space. This property allows the Server with a pre-trained model to provide recommendation services without a tedious re-training process, which significantly improves the efficiency performance. Note that the Client's data is not in the Server's database which is used for model training. In order to make the RaaS secure 2PC protocol more complete, we also precisely define the security notion of CryptoRec in this paper.

CryptoRec is able to produce recommendations in a direct mode (using only a pre-trained model learned on the Server's database which does not contain the Client's data) or in a re-training mode (where the model is first re-trained with the Client's input before computing recommendations). The re-training mode produces slightly more accurate predictions. In the direct mode, we can instantiate our protocol with an additive HE scheme such as the very efficient Paillier cryptosystem~\cite{paillier1999public}. We test both modes of CrytoRec on MovieLens-1M (ml1m)\cite{data:ml1m}, Netflix (netlfix)~\cite{data:netflix} and Yahoo-R4 (yahoo)~\cite{data:yahoo} public datasets. Experiment results show that the direct mode allows the Server with thousands of items to privately answer a prediction query in a few seconds on a single PC. To re-train the model with the Client's input, we need a limited number of homomorphic additive and multiplicative operations. Therefore, we must rely on a Somewhat HE scheme (SWHE)~\cite{fan2012somewhat}. Besides the advantage that our solution relies only on linear operations and converges in a very few numbers of iterations, the accuracy of the predictions produced by our model is less than 2\% away from those achieved by the most accurate collaborative learning algorithms known to date (depending on the datasets). In practice, the Client can choose either of the two modes, according to her preference on the trade-off between accuracy and efficiency.

%% file: preliminary.tex
\section{Preliminaries}
\label{sec:preli}

In this section, we introduce collaborative filtering and homomorphic encryption. For simplicity, the adopted notation and commonly used variables are summarized in Table~\ref{tab:notation}.  Scalars are denoted in  lower-case characters, vectors are denoted in \textbf{lower-case bold} characters, matrices are denoted in \textbf{ Upper-case bold} characters. We write $a \gets b$ to denote the algorithmic action of assigning the value of $b$ to the variable $a$, and $x \sample X$ for the action of sampling a uniformly random element $x$ from set $X$.

\begin{table}[ht]
\centering
\begin{tabular}{l|l}
\hline
$n/m$                            & number of users / items                \\
$ \textbf{R} \in \mathbb{N}^{n\times m}$   & rating matrix                          \\
$r_{ui}$                         & rating given by user $u$ for item $i$           \\
$\hat{r}_{ui}$                   & estimation of $r_{ui}$                 \\
$ \textbf{r}_u \in \mathbb{N}^{1\times m}$ & rating vector of user $u: \{ r_{ui} \}_{i=1}^{m}$           \\
$\textbf{r}_i \in \mathbb{N}^{n\times 1}$ & rating vector of item $i:\{ r_{ui} \}_{u=1}^{n} $              \\
$\bar{r}_u / \bar{r}_i$           & mean rating of user $u$ / item $i$     \\
$\phi_{ui}$                         & $\phi_{ui} =1 $ if $r_{ui}$ exists, otherwise $\phi_{ui}=0$. \\ %$\boldsymbol{\phi}_u = \{\phi_{ui} \}_{i=1}^m$        \\
\textbf{$\Theta$}                         & general form of model parameters       \\
$\llbracket x \rrbracket$          & encryption of $x$                      \\
$\llbracket \textbf{x} \rrbracket$          &  $ \{ \llbracket x_1 \rrbracket, \llbracket x_2 \rrbracket,\llbracket x_3 \rrbracket, \cdots  \}$               \\
$\pk/\sk$                          & public key / secret key                \\
%$\HEEnc()$   & homomorphic encryption: $\llbracket x \rrbracket \leftarrow \HEEnc(x,\pk)$ \\
%$\HEDec()$   & homomorphic decryption: $ x \leftarrow \HEDec(\llbracket x \rrbracket,\sk)$ \\
$\oplus$                         & addition between two ciphertexts       \\
%$\ominus$                         & subtraction between two ciphertexts (i.e., -$\oplus$)       \\
                                 & or a plaintext and a ciphertex         \\
$\odot$                          & multiplication between                 \\
                                 & a plaintext and a ciphertex            \\
$\otimes$                        & multiplication between two ciphertexts \\
$\llbracket \textbf{x} \rrbracket \llbracket \textbf{y} \rrbracket $            &  $ \sum_i \llbracket x_i\rrbracket \otimes \llbracket y_i\rrbracket $, homomorphic dot-product             \\
$ x \odot \llbracket \textbf{y} \rrbracket $            &  $ \{ x\odot \llbracket y_1 \rrbracket, x\odot \llbracket y_2 \rrbracket,x\odot \llbracket y_3 \rrbracket, \cdots  \}$             \\
$\mathcal{D}_0 \overset{s}{\approx} \mathcal{D}_1$ & distributions are statistically indistinguishable \\
$\mathcal{D}_0 \overset{c}{\approx} \mathcal{D}_1$ & distributions are computationally indistinguishable \\ \hline
\end{tabular}

\caption{Variables and notations}
\label{tab:notation}
\end{table}

\subsection{Collaborative Filtering}
\label{sec:cf}
Collaborative filtering (CF) algorithms are the current state-of-the-art recommender systems in terms of accuracy~\cite{su2009survey,zhang2017deep}. Chronologically, the memory-based approach, such as neighborhood-based methods (NBMs), is the first class of prevalent recommendation algorithms. Not only this class benefits from a good prediction accuracy, it also inherits recommendation explainability~\cite{desrosiers2011comprehensive}; Later, model-based approach algorithms, of which matrix factorization (MF) is a notable example, became popular due to their exceptional accuracy~\cite{su2009survey}. Recently, with the rise of deep learning, neural network based recommender systems are emerging, such as AutoRec~\cite{sedhain2015autorec}, leading the current prediction accuracy benchmarks~\cite{zhang2017deep}.

\subsubsection{Neighborhood-Based Method (NBM)}
The neighborhood based method estimates a user's rating on a targeted item by taking the weighted average of a certain number of ratings of the user or of the item. Formally, an item-based NBM (I-NBM) is defined as
\begin{equation}
\label{eq:i-nbm}
\hat{r}_{ui}=\bar{r}_i + \frac{\sum_{j\in \mathcal{N}_u(i)}s_{ij}(r_{uj}-\bar{r}_j)}{\sum_{j\in \mathcal{N}_u(i)}|s_{ij}|},
\end{equation}
where $\bar{r}_i$ is the mean rating of item $i$, $s_{ij} \in \textbf{S}^{m \times m}$ represents the similarity between item $i$ and $j$, and $\mathcal{N}_u(i)$ denotes a set of items rated by user $u$ that are the most similar to item $i$ according to the similarity matrix $\textbf{S} \in \mathbb{R}^{m\times m}$. Pearson correlation is one of the most widely used similarity metrics~\cite{su2009survey}:
\begin{equation}
\label{eq:pcc}
s_{ij}=\frac{\sum_{u \in \mathcal{U}_{ij}} (r_{ui}-\bar{r}_i)(r_{uj}-\bar{r}_j)}{\sqrt{\sum_{u \in \mathcal{U}_{ij}}(r_{ui}-\bar{r}_i)^2}\sqrt{\sum_{u \in \mathcal{U}_{ij}}(r_{uj}-\bar{r}_j)^2}},
\end{equation}
where $\mathcal{U}_{ij}$ denotes the set of users that rated both items $i$ and $j$.  The matrix of similarities \textbf{S} is the model parameters \textbf{$\Theta$} of NBM, $\textbf{$\Theta$} =\{ \textbf{S}\} $. User-based NBM (U-NBM) is the symmetric counterpart of I-NBM. Normally, I-NBM is more accurate and robust than U-NBM~\cite{su2009survey}.

\subsubsection{Matrix Factorization (MF)}
Let $\textbf{R}^{n\times m}$ be a sparse rating matrix formed by $n$ users and $m$ items,  in which each user rated only a small number of the $m$ items, and the missing values are marked with zero. Matrix factorization (MF) decomposes the rating matrix $\textbf{R}$ into two low-rank and dense feature matrices~\cite{koren2009matrix}:
\begin{equation}
\textbf{R} \approx  \textbf{P}\textbf{Q}^T,
\end{equation}
where $\textbf{P} \in \mathbb{R}^{n\times d}$ is the user feature space, $\textbf{Q} \in \mathbb{R}^{m\times d}$ is the item feature space and $d \in \mathbb{N}^+$ is the dimension of user and item features. To predict how user $u$ would rate item $i$, we compute $\hat{r}_{ui} = \textbf{p}_u\textbf{q}_i^T$,
where $\textbf{p}_u^{1\times d} \subset \textbf{P}$ and $\textbf{q}_i^{1\times d} \subset \textbf{Q}$ denote the learned features vectors of user $u$ and item $i$, respectively. A standard way of optimizing $\textbf{P}$ and $\textbf{Q}$ is to minimize the regularized squared error function
\begin{equation}
\label{eq:minmf}
\overset{min}{\textbf{P},\textbf{Q}} \sum_{(u,i)\in \textbf{R}}(\textbf{p}_u\textbf{q}_i^T-r_{ui})^2 + \lambda(||\textbf{p}_u||^2+||\textbf{q}_i||^2),
\end{equation}
by using the stochastic gradient descent (SGD) optimization method \cite{koren2009matrix}, but only based on observed ratings (rating matrix $\textbf{R}$ is sparse). The constant $\lambda$ is a regularization factor. The model parameters of MF are $\Theta = \{\textbf{P}, \textbf{Q}\}$.

\subsubsection{Neural Network Approach}

In addition to the success of neural networks in visual recognition and speech synthesis tasks is widely diffused, many works also focus on constructing neural recommender systems. (We refer to the reader to~\cite{zhang2017deep} for an overview.) AutoRec~\cite{sedhain2015autorec} is a notable example, built on top of Autoencoders~\cite{ng2011sparse}. Item-based AutoRec (I-AutoRec) reconstructs the inputs $\textbf{r}_{i}$ by computing
\begin{equation}
\hat{\textbf{r}}_i= f(\textbf{W}\cdot g(\textbf{V} \textbf{r}_i+\textbf{b}^{(1)})+\textbf{b}^{(2)}),
\end{equation}
where $g(\cdot)$ and $f(\cdot)$ are  activation functions, e.g. the Sigmoid function ($\frac{1}{1+e^{-x}}$) or ReLu function ($max(0,x))$. Non-linear activation functions are crucial to the success of neural networks. Model parameters are defined as follows: $\Theta = \{ \textbf{W},  \textbf{V},  \textbf{b}^{(1)},  \textbf{b}^{(2)} \}$,  where $\textbf{W} \in \mathbb{R}^{n\times d}$ and $\textbf{V} \in \mathbb{R}^{d\times n}$ are for `transformations', and $\textbf{b}^{(1)} \in \mathbb{R}^{d\times 1}$ and $\textbf{b}^{(2)} \in \mathbb{R}^{n \times 1}$ are for ``bias" terms. $\Theta $ is learned by using the SGD to minimize the regularized square error function
\begin{equation}
\label{eq:minautorec}
\overset{min}{\textbf{W},\textbf{V}, \textbf{b}^{(1)}, \textbf{b}^{(2)}} \sum_{i\in R}||\hat{\textbf{r}}_{i}-\textbf{r}_{i}||^2 + \lambda(||\textbf{W}||^2+||\textbf{V}||^2),
\end{equation}
where the gradient of each model parameter is computed by only observed ratings~\cite{sedhain2015autorec}.  Equation~(\ref{eq:minautorec}) defines I-AutoRec. The user-based AutoRec (U-AutoRec) is defined symmetrically in the obvious way. Experimental results show that I-AutoRec outperforms U-AutoRec in terms of accruracy~\cite{sedhain2015autorec}.

\subsection{Homomorphic Encryption}
Homomorphic encryption ($\HE$) is a form of encryption that allows computations to be carried over ciphertexts. The result, after decryption, is the same as if the operations had been performed on the plaintexts~\cite{gentry2009fully}. As an illustrative example, consider two plaintexts $x_1$ and $x_2$ and their corresponding ciphertexts $\llbracket x_1 \rrbracket \sample \HEEnc(x_1,$ $\pk)$ and $\llbracket x_2 \rrbracket \sample \HEEnc(x_2,\pk)$. An encryption scheme is additively homomorphic if it satisfies
$x_1+x_2=\HEDec(\llbracket x_1 \rrbracket \oplus \llbracket x_2 \rrbracket, \sk)$
or multiplicatively homomorphic if we have
$x_1\times x_2=\HEDec(\llbracket x_1 \rrbracket \otimes  \llbracket x_2 \rrbracket, \sk)$,
where $\oplus$ and $\otimes$ represent the homomorphic addition and homomorphic multiplication operations, respectively.

Some HE schemes are only either additively homomorphic or multiplicatively homomorphic, such as~\cite{paillier1999public}. The schemes that fall into this category are know to be {\em partially homomorphic} (PHE). Schemes that support both additions and multiplications, but only a limited number of times, are known as \emph{somewhat homomorphic} (SWHE), as opposed to those that allow an unbounded number of homomorphic operations, which are called \emph{fully homomorphic encryption} (FHE) schemes~\cite{gentry2009fully,fan2012somewhat}.
The efficiency of the schemes in each class is usually related to the expressiveness of the supported operations, meaning that PHE schemes are more efficient than SWHE schemes, which in turn are more efficient that FHE schemes.

In addition to the additively or multiplicatively homomorphic properties of ciphertexts, HE schemes also allow additions and multiplications between a ciphertext and a plaintext, i.e. $x_1+x_2 = \HEDec(\llbracket x_1 \rrbracket \oplus x_2,\sk)$  and $x_1 \times x_2 = \HEDec(\llbracket x_1 \rrbracket \odot x_2,\sk)$.

\heading{Syntax} A $\HE$ scheme is a tuple of four ppt algorithms $\HE := (\HESetup, \HEEnc, \HEEval, \HEDec)$ as follows:
\begin{itemize}
\item $\HESetup(\secpar)$ is the setup algorithm. It takes as input the security parameter $\lambda$ and outputs a private/public key pair $(\sk,\pk)$. The public key $\pk$ includes a description of the message space $\mathcal{M}$.
\item $\HEEnc(\msg, \pk)$ is the encryption algorithm, which takes as input the public key $\pk$ and a message $\msg \in \mathcal{M}$ and outputs a ciphertext $\ctxt$.
\item $\HEEval(f,\ctxt_1,...,\ctxt_t,\pk)$ is the homomorphic evaluation algorithm. It takes as input a public key $\pk$, a circuit $f : \mathcal{M}^t \rightarrow \mathcal{M}$ in a class $\mathcal{F}$ of supported circuits and $t$ ciphertexts $\ctxt_1,...,\ctxt_t$, and returns a ciphertext $\ctxt$.
\item $\HEDec(\ctxt,\sk)$ is the decryption algorithm that on input a secret key $\sk$ and a ciphertext $\ctxt$, it returns a message $\msg$ or a special failure symbol $\perp$.
\end{itemize}
We now briefly describe correctness, IND-CPA security and circuit privacy for a homomorphic encryption scheme $\HE$.

\heading{Correctness} $\HE$ is correct if for all honestly generated keys $(\pk,\sk) \sample \HESetup(\secpar)$, for all supported $f \in \mathcal{F}$ and for all messages $(\msg_1,...,\msg_t) \in \mathcal{M}^t$, we have that if $\ctxt_i \sample \HEEnc(\pk,\msg_i)$, $\forall i \in [t]$, then it holds with overwhelming probability over the random coins of all algorithms that
\[
\HEDec(\HEEval(f,(\ctxt_1,...,\ctxt_t),\pk),\sk) = f(\msg_1,...,\msg_t).
\]

\heading{IND-CPA security}
This is the standard notion of security for any homomorphic encryption scheme. It guarantees that nothing can be learned from a ciphertext about the message it encrypts (beyond, perhaps, its length).

\begin{definition}
We say that $\HE$ is $\INDCPA$ secure if for every {\em legitimate} ppt adversary $\A := (\A_0,\A_1)$ the following definition of advantage is negligible in the security parameter:
\[
\Adv^{\indcpa}_{\HE,\A}(\lambda) := 2 \cdot\Pr[\INDCPA_\HE^\A(\secpar)] - 1\,,
\]
where game $\INDCPA_\HE^\A$ is described in Figure.~\ref{fig:def:indcpa} and a legitimate adversary outputs in its first stage (i.e. algorithm $\A_0$) two messages of equal bit length.
\end{definition}

\begin{figure}[ht]
	\centering\framebox{
	\begin{tabular}{lll}
		\begin{varwidth}[t]{1\textwidth}
			\underline{$\INDCPA_{\HE}^{\A}(\secpar)$}: \\
			$(\sk,\pk) \sample \HESetup(\secpar)$ \\
			$(\msg_0,\msg_1,\st) \sample \A_0(\secpar,\pk)$ \\
			$b \sample \{0,1\}$ \\
			$\ctxt \sample \HEEnc(\pk,\msg_b)$ \\
			$b' \sample \A_1(\st,\ctxt)$ \\
			return $(b = b')$
		\end{varwidth}
	\end{tabular}
	} \caption{$\INDCPA$ security game} \label{fig:def:indcpa} \label{fig:def:indcpa}
\end{figure}

\heading{Circuit privacy} An additional requirement of many HE applications, including ours, is that the evaluated ciphertext should also hide the function $f$, apart from what is inevitably leaked through the outcome of the computation. This property is known as the circuit privacy~\cite{Gentry2009,ip2007}.

\begin{definition}
\label{def:circuitpriv}
A homomorphic encryption scheme $\HE$ is circuit private if there exists a ppt simulator $\Sim$ such that for any security parameter $\lambda$, any key pair $(\sk,\pk) \sample \HESetup(\secpar)$, any supported function $f \in \mathcal{F}$ and any tuple of messages $\msg_1,...,\msg_t \in \mathcal{M}^t$, it holds that
\begin{flalign*}
((\sk,\pk),(\ctxt_1,...,\ctxt_t),\HEEval(f,\ctxt_1,...,\ctxt_t,\pk)) \\
\overset{s}{\approx} ((\sk,\pk),(\ctxt_1,...,\ctxt_t),\Sim(f(\msg_1,...,\msg_t),\pk)),
\end{flalign*}
where $\ctxt_i \sample \HEEnc(\pk,\msg_i)$.
\end{definition}

%% file: security.tex
\section{Security Guarantees}
\label{sec:security}
We consider a general protocol between two parties, referred here as Client and Server. The Client has input $x$, while the Server's input consists of the ML algorithm $f$, random coins $r$, and the model parameters pre-trained with the training set $y$. The protocol, which in general can involve multiple rounds of interaction but in our case (i.e., CryptoRec) is simply non-interactive, is denoted by $\mathrm{Client}(x)\leftrightarrow\mathrm{Server}(f,y,r)$. We chose to give $r$ as an explicit input to the Server as these coins are used by the ML algorithm $f$, if probabilistic. Nevertheless, note that both Client and Server are (possibly) randomized algorithms as well and might use additional random coins.
The goal is for the Client to learn the prediction $f(x,y;r)$, so the protocol is complete if the final output of the Client $\mathsf{out}_\mathrm{Client}[\mathrm{Client}(x)\leftrightarrow\mathrm{Server}(f,y,r)] = f(x,y;r)$.
Security demands that the Client learns nothing on the Server's input beyond what can be inferred from the Client's input $x$ and the outcome prediction $f(x,y;r)$. It also demands that the Server learns nothing on the Client's input $x$. Formally, we adopt the real/ideal paradigm:

\begin{itemize}
\item{\bf Real World.} The protocol $\Client(x)\leftrightarrow\Server(f,y,r)$ is executed and the real-world adversary $\A$ can corrupt either the Client or the Server (but not both at the same time), meaning $\A$ can choose the input of the corrupted party and observe all communications. We assume that the adversary is static and semi-honest, i.e. the adversary chooses which party to corrupt before the protocol execution starts and the corrupted party follows the protocol honestly.
\item{\bf Ideal World.} The Client and Server send their input $x$ and $(f,y,r)$, respectively, to a trusted third party (TTP) that computes $f(x,y;r)$. The Client gets the output of the computation $f(x,y;r)$ while Server gets nothing.
\end{itemize}

\begin{definition}
For every adversary $\A$, there exist ppt simulators $\Sim_1$ and $\Sim_2$ such that for all inputs $x$ and $(f,y,r)$: \\

$\mathsf{view}_\Client[\A(x)\leftrightarrow\Server(f,y,r)] \overset{c}{\approx} \Sim_1(x,f(x,y;r))$, \\
\indent $\mathsf{view}_\Server[\Client(x)\leftrightarrow\A(f,y,r)] \overset{c}{\approx} \Sim_2(|x|,f,y,r)$.
\label{def:sec}
\end{definition}

%Any ``prior knowledge'' that the parties have about each others' inputs is captured by the non-uniformity of the simulator and implicit distinguisher.

We derive a simple non-interactive 2PC protocol from $\HE$ and show that it provides the aforementioned security guarantees. Let $\Pi$ be a non-interactive protocol, where the Client generates its own public key pair $(\sk,\pk) \sample \HESetup(1^\lambda)$, encrypts its input $\ctxt \sample$ $\HEEnc(x,\pk)$ and sends $\ctxt$ to the Server. The Server evaluates $\ctxt' \sample \HEEval(f(\cdot,y;r),\ctxt,\pk)$ and sends it back to the Client. The Client decrypts $\ctxt'$ and outputs the prediction $p \get \HEDec(\ctxt',\sk)$. %Let $\HE$ be a Homomorphic Encryption scheme %homomorphically

\begin{theorem}
\label{theorem:security}
If $\HE$ is IND-CPA-secure and circuit private, then $\Pi$ satisfies security definition~\ref{def:sec}.
\end{theorem}
\begin{proof}
This proof is divided into two parts: We show by means of a direct reduction to the IND-CPA and circuit privacy properties of $\HE$ that $\Pi$ is private against the Server and against the Client.

The client's view $\mathsf{view}_\Client[\Client(x)\leftrightarrow\Server(f,y,r)]$ consists of tuples of the form $(\sk,\pk,x,\ctxt,\ctxt')$.
Let $\Sim$ be a simulator for circuit privacy [Definition~\ref{def:circuitpriv}]. Let $\Sim_1(\sk,\pk,x,$ $f(x,y;r))$ be an algorithm that uses $\Sim(f(x,y;r),\pk))$ as a subroutine to produce $\ctxt'$ and outputs $(\sk,\pk,x,\HEEnc(x,\pk),$ $\ctxt')$.
Suppose that $\Pi$ is not secure and that there is a successful distinguisher $\mathcal{D}_1$ that distinguishes $\mathsf{view}_\Client[\A(x)\leftrightarrow\Server(f,y,r)]$ from $\Sim_1(\sk,\pk,x,f(x,y;r))$ with non-negligible probability.
If $\mathcal{D}_1$ exists, then a successful distinguisher $\mathcal{D}$ against simulator $\Sim$ also exists, which contradicts the circuit privacy property of $\HE$: Distinguisher $\mathcal{D}(\sk,\pk,\ctxt,\ctxt')$ decrypts $x \get \HEDec(\ctxt,\sk)$, runs $\mathcal{D}_1(\sk,\pk,x,\ctxt,\ctxt')$ as a subroutine, and outputs whatever $\mathcal{D}_1$ outputs, winning whenever $\mathcal{D}_1$'s guess is correct.

The server's view $\mathsf{view}_\Server[\Client(x)\leftrightarrow\Server(f,y,r)]$ consists of tuples of the form $(\pk,\ctxt)$.
Let $\Sim_2(\pk,|x|)$ be the algorithm that samples $x' \sample \{0,1\}^{|x|}$, encrypts $\ctxt \sample$ $\HEEnc(x',\pk)$ and returns $(\pk,\ctxt)$.
Suppose that $\Pi$ is not secure and that there is a successful distinguisher $\mathcal{D}_2$ that distinguishes $\mathsf{view}_\Server[\Client(x)\leftrightarrow\A(f,y,r)]$ from $\Sim_2(\pk,\ctxt)$ with non-negligible probability. Then, we can construct an adversary $\mathcal{B}$ against the IND-CPA property of $\HE$ as follows: $\mathcal{B}$ gets $x$ from $\Client$ and samples $x' \sample \{0,1\}^{|x|}$ as $\Sim_2$. $\mathcal{B}$ then asks to be challenged on the pair $(x,x')$ and gets the ciphertext $\ctxt$ that encrypts either $x$ or $x'$. $\mathcal{B}$ runs distinguisher $\mathcal{D}_2(\pk,\ctxt)$ as a subroutine and outputs $\mathcal{D}_2$'s guess as its own. $\mathcal{B}$ perfectly simulates for $\mathcal{D}_2$ either the real-world or the ideal-world of simulator $\Sim_2$, winning the IND-CPA game whenever $\mathcal{D}_2$'s guess is correct.
\end{proof}

%We take a moment to make a few remarks on our security model. 
\emph{Remarks on our security model}. First, the security guarantees extend those of 2PC, as {\it circuit privacy} allows us to hide the ML algorithm $f$ from the Client, in addition to protecting the pre-trained model $y$ and random coins $r$. Secondly, if the underlying $\HE$ scheme is circuit-private against malicious adversaries~\cite{OPP2014} then our protocol automatically satisfies a definition of similar ``flavor'' and becomes secure against malicious adversaries. Note that in our CryptoRec protocol the Client talks first and $\mathsf{view}_\Server[\Client(x)\leftrightarrow\A(f,y,r)]$ is the same regardless $\A$ is malicious or semi-honest. Therefore, we only need to consider security against malicious adversaries for circuit privacy. Third, we cannot guarantee the predictions that the Server computes are correct (In fact, the Server has incentives to compute predictions correctly). Security guarantees that nothing about the Client's input is revealed to the Server and nothing about the Server's prediction model is revealed to the Client beyond what can be inferred by the prediction. Lastly, since we rely on a standard definition of circuit privacy, which pertains to the case of a single use of each public key, the Client needs to generate a new key-pair for each query.

%% file: newmodel.tex
\section{CryptoRec}
\label{sec:cryptorec0}
In this section, we present CryptoRec, a non-interactive secure 2PC protocol built on top of a new homomorphic encryption-friendly recommender system, referred to as CryptoRec's model. In Section~\ref{sec:cfrs}, we introduce CryptoRec's model. In Section~\ref{sec:training} we explain how to train the model and learn the parameters $\Theta$. Finally, in Section~\ref{sec:protocols}, we combine the prediction procedure of our model with homomorphic encryption. This gives rise to our CryptoRec protocol. We also consider a second variant of the protocol in which the model parameters $\Theta$ are retrained before computing recommendations. The re-training occurs in encrypted form, therefore it results in better accuracy without compromising security. Naturally, the computational cost on the Server side is considerably heavier.

\subsection{CryptoRec's Model}
\label{sec:cfrs}
Existing collaborative filtering (CF) technologies require non-linear operations or re-training with the Client's data~\cite{zhang2017deep,koren2009matrix,mnih2008probabilistic,zhang2017autosvd}. Directly applying the existing CFs to encrypted data leads to severe efficiency problems. To address this issue, we propose CryptoRec's model, a new homomorphic encryption friendly recommender system. It models user-item interaction behaviors in an item-only latent feature space. This means that the user features do not exist in the latent feature space, the model will automatically compute the user features by aggregating pre-learned item features. This property allows the Server with a pre-trained model to provide recommendations for the Client without having to re-train the model with the Client's data. Algebraic operations in CryptoRec's model are constrained to only additions and multiplications, thus CryptoRec's model is straightforwardly compatible with homomorphic encryption schemes.

We exploit the fact that a user profile is essentially identified by items that the user has rated, to construct personalized user features in an item-only latent feature space. In particular, we model the personalized user features $\textbf{p}_u$ by aggregating pre-learned item features $\textbf{Q}= \{ \textbf{q}_i\}_{i=1}^{m}$ as follows,
\begin{equation}
\label{eq:i2u}
\textbf{p}_u=  \textbf{r}_u\textbf{Q}
\end{equation}
therefore we can approximate an observed rating $r_{ui}$ by
\begin{equation}
\label{eq:ca1}
r_{ui} \approx \hat{r}_{ui} = \underbrace{(\textbf{r}_u\textbf{Q})}_{\textbf{p}_u}\textbf{q}_i^T
\end{equation}

 Using only a single latent feature space $\textbf{Q}$ to model a large number of ratings often leads to an information bottleneck. To address this issue, we relax the item features  which were used to construct user features $\textbf{P}$, and redefine the Equation (\ref{eq:ca1}) as
\begin{equation}
\label{eq:ca2}
r_{ui} \approx  \underbrace{(\textbf{r}_u\textbf{A})}_{\textbf{p}_u}\textbf{q}_i^T
\end{equation}
Note that $\textbf{A} \in \mathbb{R}^{m\times d}$ is a new item feature space.

We now have the basic form of CryptoRec's model which has an item-only latent feature space and relies only on addition and multiplication operations. However, it is not robust enough in practice due to the high variance of individual user or item behaviors,  commonly known as biases. For example, real-world datasets exhibit large systematic tendencies for some users to give higher ratings than others, and for some items to receive higher ratings than others~\cite{koren2009matrix}. To address this issue, a common approach is to identify the portion of these ratings that individual user or item biases can explain, subjecting only the true interaction portion of the ratings to latent factor modeling~\cite{koren2009matrix}. The user and item biases are often approximated by
\begin{equation}
\label{eq:ebb}
b_{ui} = \mu+b_u+b_i
\end{equation}
where $\mu = \frac{\sum_{(u,i)\in \textbf{R}}r_{ui}}{N}$ is the global rating average, $N$ is the number of observed ratings. $b_u $ and  $b_i$  approximate user bias and  item biases, respectively. To obtain $b_u$ and $b_i$, we can either compute $b_u = \bar{r}_u - \mu$ and $b_i = \bar{r}_i - \mu$~\cite{koren2009matrix}, or directly learn their values from a dataset~\cite{koren2008factorization}. The former ignores the global effects upon a single user or item; the latter models both the individual behaviors and global effects, but sometimes it leads to an early overfitting. To maintain both reliability and accuracy, we separately model the individual behaviors and global effects as follows
\begin{equation}
\label{eq:ebb2}
b_{ui} = \mu+b_u+b_i+b_u^*+b_i^*
\end{equation}
where $b_u$ and $b_i$ are computed as  $b_u = \bar{r}_u - \mu$ and $b_i = \bar{r}_i - \mu$. $b_u^*$ and $b_i^*$ are the parameters directly learned from the dataset to capture only the impact of global effects upon a single user and item, respectively.

We combine the biases approximator (Equation (\ref{eq:ebb2})) and the user-item interaction approximator (Equation (\ref{eq:ca2})) to formalize the final CryptoRec's model as following,
\begin{equation}
\label{eq:ebm1}
r_{ui} \approx \hat{r}_{ui} = \underbrace{\mu+b_u+b_i+b_u^*+b_i^*}_{biases} +\underbrace{ (\textbf{r}_u\textbf{A})\textbf{q}_i^T}_{interaction}
\end{equation}
As such, the user preference estimation is separated into two parts: biases approximator and user-item interaction approximator. This allows only the true user-item interaction being modeled by the factor machine (i.e., Equation (\ref{eq:ca2})).  The model parameters of CryptoRec's model are $\Theta = \{\textbf{A}, \textbf{Q},  \textbf{b}_u^*, \textbf{b}_i^* \}$ \footnote{For the convenience of notation in describing the algorithms later on, we omit  $\{\mu, \textbf{b}_u, \textbf{b}_i \} $ from the model parameters $\Theta$  in favor of a slightly more succinct notation.  Note that $\{\mu, \textbf{b}_u, \textbf{b}_i \} $ are not learned by training procedure either. }, where $\textbf{b}_u^* = \{b_u^* \}_{u=1}^n$ , $\textbf{b}_i^* = \{b_i^* \}_{i=1}^m$.

\subsection{Training}
\label{sec:training}
The model parameters $\Theta = \{\textbf{A},\textbf{Q}, \textbf{b}_u^*, \textbf{b}_i^* \}$ are learned by solving the regularized least squares objective function,
\begin{equation}
\label{eq:obj}
\begin{split}
\mathcal{L} = & \sum_{u=1}^{n}  ||(\hat{\textbf{r}}_{u}-\textbf{r}_{u})\cdot \boldsymbol{\phi}_u||^2 \\ + & \lambda\cdot (\left \| \textbf{A}\right \|^2 + \left \|\textbf{Q}  \right \|^2  + \left \| \textbf{b}_u^* \right \|^2 + \left \| \textbf{b}_i^* \right \|^2)
\end{split}
\end{equation}
where $\hat{r}_{ui}$ is defined in Equation (\ref{eq:ebm1}),  $\boldsymbol{\phi}_u = \{\phi_{ui} \}_{i=1}^m$ and  $(\hat{\textbf{r}}_{u}-\textbf{r}_{u})\cdot \boldsymbol{\phi}_u$ denotes $\{(\hat{r}_{ui}-r_{ui})\phi_{ui} \}_{i=1}^m$. If user $u$ rated item $i$, then $\phi_{ui} =1$, otherwise, we let $\phi_{ui}=0$ and  $r_{ui}=0$. We use $\phi_{ui}$ to remove the (incorrect) gradients computed on unobserved ratings. The constant $\lambda$ controls the extent of regularization. When performing training on plaintext dataset, the Server can compute $\phi_{ui}$ by itself, avoiding the unnecessary gradient computations on unobserved ratings.

As shown in Equation (\ref{eq:ebm1}), CryptoRec's model is in fact a two-layer network. The first layer outputs the user feature vector $\textbf{p}_u=\textbf{r}_u\textbf{A}$ and the second layer integrates the user features, item features and biases to estimate the user preferences. Back-propagation~\cite{lecun1990handwritten} is a standard method used in neural networks to calculate a gradient that is needed in the calculation of the values of parameters in the network. Simply put, we first compute the predictions given the input data (i.e., a forward pass); and then we calculate the total error according to the objective function (e.g., Equation (13)). Lastly, we compute the gradient of each trainable parameter using the error back-propagated through the network layers (i.e., a backward pass). Using the back-propagation method, we have the gradient of each model parameter of CrypotRec's model as follows,

\begin{equation}
\label{eq:grad}
\begin{split}
& \Delta \textbf{A} = \frac{\partial \mathcal{L}  }{\partial \textbf{p}_u }\cdot\frac{\partial \textbf{p}_u  }{\partial \textbf{A} } =    [(\textbf{e}_u \cdot \boldsymbol{\phi}_u )\textbf{Q}] \circledast \textbf{r}_u^T +\lambda \cdot \textbf{A}  \\
& \Delta \textbf{q}_i = \frac{\partial \mathcal{L}  }{\partial \textbf{q}_i} =  \phi_{ui}\cdot (e_{ui}\cdot(\textbf{r}_u\textbf{A})+\lambda \cdot \textbf{q}_i) \\
& \Delta b_u^* = \frac{\partial \mathcal{L}  }{\partial b_u^* } =   \textbf{e}_u \boldsymbol{\phi}_u + \lambda  \cdot b_u^*\\
& \Delta b_i^* = \frac{\partial \mathcal{L}  }{\partial b_i^* } =  \textbf{e}_{i}\boldsymbol{\phi}_{i}+ \lambda  \cdot b_i^*\\
\end{split}
\end{equation}
where $e_{ui} = \hat{r}_{ui}-r_{ui}$,  $\textbf{e}_u = \{e_{ui} \}_{i=1}^m$ ,  $\textbf{e}_i = \{e_{ui} \}_{u=1}^n$, and $\textbf{e}_u \cdot \boldsymbol{\phi}_u=\{e_{ui}\cdot\phi_{ui} \}_{i=1}^m$. $\circledast $ denotes outer product \footnote{Given two vectors $\textbf{x}^{1\times m}$ and $\textbf{y}^{n\times 1}, ( \textbf{x} \circledast  \textbf{y})_{ij}=x_iy_j$ }.  We randomly divide the dataset into multiple batches. In the training phase, we compute gradient by batch and update the model parameters by moving in the opposite direction of the gradient (i.e., gradient descent optimization algorithm). Algorithm \ref{alg:train} outlines the model training procedure. The learning rate $\eta$ is used to control the speed of model updates. Note that the training procedure only relies on addition and multiplication operations.

\begin{algorithm}[ht]
\caption{CryptoRec's model training procedure $\mathcal{T}$ }
\label{alg:train}
\begin{flushleft}
 \textbf{Input:   }  Rating  $\textbf{R}$, rating indicator $\boldsymbol{\Phi} $ , user mean ratings $\bar{\textbf{r}}_u = \{\bar{r}_u \}_{u=1}^n$,   $\Theta = \{\textbf{A}^{(0)},\textbf{Q}^{(0)}, \textbf{b}_u^{*(0)}, \textbf{b}_i^{*(0)} \}$ \\
 \textbf{Output:}   Optimized   $\Theta = \{\textbf{A}^{(K)}, \textbf{Q}^{(K)},  \textbf{b}_u^{*(K)}, \textbf{b}_i^{*(K)} \}$
\end{flushleft}
\algrule
\begin{algorithmic}[1]
\Procedure{$\mathcal{T}$}{$ \{ \textbf{R},  \boldsymbol{\Phi}, \bar{\textbf{r}}_u \}, \Theta$}
\For{\texttt{$k \gets \{ 1,2,\cdots,K \}$ }}
        \State $\textbf{A}^{(k)} \gets \textbf{A}^{(k-1)} -  \eta \cdot \Delta \textbf{A}^{(k-1)}$ \Comment{$\eta$: learning rate}
        \State $\textbf{Q}^{(k)} \gets \textbf{Q}^{(k-1)} -  \eta \cdot \Delta \textbf{Q}^{(k-1)}$
        \State $\textbf{b}_u^{*(k)} \gets \textbf{b}_u^{*(k-1)}  -  \eta \cdot \Delta \textbf{b}_u^{*(k-1)} $
        \State $\textbf{b}_i^{*(k)} \gets \textbf{b}_i^{*(k-1)} -  \eta\cdot \Delta \textbf{b}_i^{*(k-1)} $
\EndFor
%\EndFor
\State \textbf{return} $\Theta = \{\textbf{A}^{(K)}, \textbf{Q}^{(K)}, \textbf{b}_u^{*(K)}, \textbf{b}_i^{*(K)} \}$
\EndProcedure
\end{algorithmic}
\end{algorithm}

\subsection{Two  Secure Protocols}
\label{sec:protocols}
In this section, we introduce two CryptoRec secure protocols.  In the first protocol, the Server uses pre-trained model parameters $\Theta$ and directly takes as input the Client's encrypted rating vector to compute recommendations. In the second protocol, the Server re-trains the  model parameters $\Theta$ before computing recommendations.  For the sake of clarity, we denote the Client as $v$ in this section.

\emph{Secure protocol with a pre-trained model}. Figure~\ref{tab:woretmodel} describes the security protocol for prediction with pre-trained model parameters $\Theta$. The Client $v$ sends $ \llbracket \textbf{r}_v \rrbracket$ and  $ \llbracket  \bar{r}_v  \rrbracket$  to the Server, which executes the prediction process $\mathcal{P}$ (described in Figure~\ref{tab:woretmodel}) and returns the encrypted results $\hat{\textbf{r}}_v$.

\begin{figure}[ht]
\centering\framebox{
\begin{tabular}{lll}
%\hline
\textbf{User}  (\textbf{Input:} $\textbf{r}_v, \bar{r}_v$)                                                             &                                                                           &\textbf{ Server}  (\textbf{ Input:} $\Theta$)                                                                                \\ \hline
$\llbracket \textbf{r}_v \rrbracket \sample \HEEnc(\textbf{r}_v,\pk)$             &                                                                           &                                                          \\
$\llbracket \bar{r}_v \rrbracket \sample  \HEEnc(\bar{r}_v,\pk)$ &                                                                           &                                           \\
$\llbracket \chi  \rrbracket \gets \{ \llbracket \textbf{r}_v \rrbracket,\llbracket \bar{r}_v \rrbracket \}$ &                                                                           &                                             \\
                                                                 & $\xrightarrow{\llbracket \chi  \rrbracket}$ &                                                                                       \\
%                                                                 &                                                                           & $\llbracket\chi_1\rrbracket  \gets \{ \llbracket \textbf{r}_u \rrbracket, \llbracket \bar{r}_u \rrbracket \}$                     \\
                                                                 &                                                                           & $\llbracket \hat{\textbf{r}}_v \rrbracket \leftarrow \mathcal{P}(\llbracket\chi \rrbracket, \Theta)$ \\
                                                                 & $\xleftarrow{\llbracket \hat{\textbf{r}}_v \rrbracket}$                             &                                                                                       \\
$\hat{\textbf{r}}_v \leftarrow \HEDec(\llbracket \hat{\textbf{r}}_v \rrbracket,\sk)$   &                                                                           &                                                                                       %\\ \hline
\end{tabular}
}
\caption{CryptoRec with pre-trained $\Theta = \{\textbf{A}, \textbf{Q}, \textbf{b}_u^{*}, \textbf{b}_i^{*} \}$}
\label{tab:woretmodel}
\end{figure}

\begin{algorithm}[ht]
\caption{CryptoRec's model prediction procedure $\mathcal{P}$ }
\label{alg:predc}
\begin{flushleft}
 \textbf{Input:   }  Ratings  $\llbracket \textbf{r}_v \rrbracket,\llbracket \bar{r}_v \rrbracket$ ,    $\Theta = \{\textbf{A}, \textbf{Q}, \textbf{b}_u^{*}, \textbf{b}_i^{*},  \mu, \textbf{b}_u, \textbf{b}_i \}$ \\
 \textbf{Output:}   Recommendations $\llbracket \hat{\textbf{r}}_v \rrbracket$
\end{flushleft}
\algrule
\begin{algorithmic}[1]
\Procedure{$\mathcal{P}$}{$\{ \llbracket \textbf{r}_v \rrbracket,\llbracket \bar{r}_v \rrbracket \}, \Theta$}
\If{$ b_v^* \notin  \textbf{b}_u^* $} \Comment{$b_v^* \in \textbf{b}_u^*$  if re-traininng the $\Theta$ }
    \State $b_v^* \gets \frac{\sum_{u=1}^{n}b_u^*}{n}$
\EndIf
\State $\llbracket  \textbf{p}_v \rrbracket \gets \llbracket \textbf{r}_v \rrbracket \textbf{A}$ \Comment{HE dot-product using $\odot$ and $\oplus$}
%\State $\llbracket b_v \rrbracket \gets  \llbracket \bar{r}_v \rrbracket \oplus (- \mu) $
\For{\texttt{$ i \gets [1,2,\cdots,m]$ }}
        \State $\llbracket x_1 \rrbracket  \gets ( b_i + b_i^*+ b_v^*) \oplus  \llbracket \bar{r}_v \rrbracket $  \Comment{$b_v  \gets   \bar{r}_v - \mu $}
        \State $ \llbracket x_2 \rrbracket \gets  \llbracket \textbf{p}_{v} \rrbracket  \textbf{q}_i^T$
        \State $\llbracket \hat{r}_{vi} \rrbracket \gets \llbracket x_1 \rrbracket  \oplus  \llbracket x_2 \rrbracket  $
        \State $ \llbracket \hat{\textbf{r}}_{v} \rrbracket[i] \gets \llbracket \hat{r}_{vi} \rrbracket $
\EndFor
\State \textbf{return} $ \llbracket \hat{\textbf{r}}_{v} \rrbracket$
\EndProcedure
\end{algorithmic}
\end{algorithm}

We present the prediction process $\mathcal{P}$ of CryptoRec's model in Algorithm \ref{alg:predc}. The computation is straightforward since CryptoRec's model contains only addition and multiplication operations, as we can observe in Equation (\ref{eq:ebm1}). The inputs of this algorithm are the Client's encrypted rating vector $ \llbracket \textbf{r}_v \rrbracket$,  the average rating $ \llbracket  \bar{r}_v  \rrbracket$ and model parameters $\Theta$. Since the Client's $b_v^*$  is unknown to the Server, $b_v^*$ is set to the average value of $\textbf{b}_u^*$ (line 2-3, Algorithm \ref{alg:predc}).

\emph{Secure protocol with re-training}. In order to achieve the most accurate predictions, we introduce a re-training process to the CryptoRec protocol, shown in Figure~\ref{tab:retmodel}. Compared to the secure protocol without a re-training step (using only a pre-trained model, Figure~\ref{tab:woretmodel}), there are two differences: The first one is that, besides $ \llbracket \textbf{r}_v \rrbracket$ and $ \llbracket  \bar{r}_v  \rrbracket$, the user also sends the encrypted indication vector $\llbracket  \boldsymbol{\phi}_v \rrbracket$ to the Server, which will be used in the training procedure $\mathcal{T}$ . The second one is that, before computing recommendations using $\mathcal{P}$ (Algorithm \ref{alg:predc}), the Server re-trains the model parameters $\Theta$ with the Client's inputs.

\begin{figure}[ht]
\hspace*{-10px}
\centering\framebox{
\begin{tabular}{lll}
\textbf{User}  (\textbf{Input:} $\textbf{r}_v, \boldsymbol{\phi}_v, \bar{r}_v$)                                                               &                                                                           & \textbf{ Server}  (\textbf{ Input:} $\Theta$)                                                                                 \\ \hline
$\llbracket \textbf{r}_v \rrbracket \sample  \HEEnc(\textbf{r}_v,\pk)$             &                                                                           &                                                          \\
$\llbracket \boldsymbol{\phi}_v \rrbracket \sample \HEEnc(\boldsymbol{\phi}_v,\pk)$     &      &  \\
$\llbracket \bar{r}_v \rrbracket \sample  \HEEnc(\bar{r}_v,\pk)$ &                                                                           &                                             \\
$\llbracket \chi  \rrbracket \gets \{ \llbracket \textbf{r}_v \rrbracket,\llbracket \boldsymbol{\phi}_v \rrbracket,\llbracket \bar{r}_v \rrbracket \}$ &                                                                           &                                             \\
                                                                 & $\xrightarrow{\llbracket \chi  \rrbracket}$ &                                                                                       \\
%                                                                 &                                                                           & $\llbracket\chi_1\rrbracket \gets \{ \llbracket \textbf{r}_u \rrbracket, \llbracket \boldsymbol{\phi}_u \rrbracket, \llbracket \bar{r}_u \rrbracket  \}$                     \\
                                                                 &                                                                           & $\llbracket \Theta \rrbracket \leftarrow \mathcal{T}(\llbracket\chi \rrbracket ,\Theta)$                    \\
                                                                 &                                                                           & $\llbracket\chi_1\rrbracket  \gets \{ \llbracket \textbf{r}_v \rrbracket, \llbracket \bar{r}_v \rrbracket \} $                     \\
                                                                 &                                                                           & $\llbracket \hat{\textbf{r}}_v \rrbracket \leftarrow \mathcal{P}(\llbracket\chi_1\rrbracket, \llbracket \Theta \rrbracket)$ \\
                                                                 & $\xleftarrow{\llbracket \hat{\textbf{r}}_v \rrbracket}$                            &                                                                                       \\
$\hat{\textbf{r}}_v \sample \HEDec(\llbracket \hat{\textbf{r}}_v \rrbracket,\sk)$   &                                                                           &                                                                                       \\ %\hline
\end{tabular}
}
\caption{CryptoRec with re-training  $\Theta = \{\textbf{A}, \textbf{Q}, \textbf{b}_u^{*}, \textbf{b}_i^{*} \}$}
\label{tab:retmodel}
\end{figure} 

%\begin{remark}[Re-training Process]
The training procedure $\mathcal{T}$ is described in Algorithm \ref{alg:train}, takes advantage of homomophic properties of the encryption scheme.  It is worth stressing that in the re-training protocol, we re-train the model parameters $\Theta$ with only the Client's data, not the Server's dataset. For efficiency reasons, the Server should pre-train the model parameters $\Theta$ on its dataset. Note that after the re-training process, the model parameters $\Theta$ are encrypted. So the related algebraic operations in the prediction process $\mathcal{P}$ should be also updated to their corresponding homomorphic operations.
%\end{remark}

It's important to notice that both protocols presented here are instances of the general non-interactive 2PC protocol of Section~\ref{sec:security}, in which the Client's input $x$ is set as $(\textbf{r}_v, \boldsymbol{\phi}_v, \bar{r}_v)$, and the Server's inputs $y$ and $r$ are set as $(\textbf{R}, \boldsymbol{\Phi}, \bar{\textbf{r}}_v)$ and as the random coins of the training procedure $\mathcal{T}$. The function $f$ that is homomorphically evaluated is set as a composition of the training procedure $\mathcal{T}$ and the prediction procedure $\mathcal{P}$, i.e. in the first CryptoRec protocol, without any re-training step, $f(x,y;r) := \mathcal{P}(x,\mathcal{T}(y;r))$ and in second CryptoRec protocol $f(x,y;(r_1,r_2)) := \mathcal{P}(x,\mathcal{T}(x,\mathcal{T}(y;r_1);r_2))$. Therefore, by Theorem~\ref{theorem:security}, the security guarantees of the generic 2PC protocol of Section~\ref{sec:security} are preserved in our protocols.

%% file: experiment.tex
\section{Experiment Setup}
We evaluate the accuracy and efficiency performances of CryptoRec on the rating prediction task and compare CryptoRec with several state-of-the-art collaborative filtering algorithms, including item-based NBM (I-NBM)~\cite{desrosiers2011comprehensive}, biased matrix factorization (BiasedMF)~\cite{koren2009matrix}, user-based AutoRec (U-AutoRec)~\cite{sedhain2015autorec} and item-based AutoRec (I-AutoRec)~\cite{sedhain2015autorec}. We test these models on three datasets which are widely used for recommender systems performance evaluation, as shown in Table \ref{tab:ds3}. The dataset ml1m~\cite{data:ml1m} contains 1 million ratings; yahoo~\cite{data:yahoo} contains 0.21 million ratings;  For netflix~\cite{data:netflix} dataset, we select 11, 000 users who have given 1.2 million ratings to 4,768 items, where each user has at least 30 ratings.  The testbed is a single PC with 8 Intel (R) Xeon(R) CPUs running at 3.5 GHz, with 32 GB of RAM, running the Ubuntu 16.04 operating system. All the 8 CPUs are used in the experiments.

\begin{table}[h]
\centering
\begin{tabular}{|l|l|l|l|l|}
\hline
        & user \# & item \# & density & scale     \\ \hline
netflix & 11,000 & 4,768  & 2.17\%  & {[}1,5{]} \\ \hline
ml1m    & 6,040   & 3,952   & 4.2\%   & {[}1,5{]} \\ \hline
yahoo   & 7,637   & 3,791   & 0.72\%  & {[}1,5{]} \\ \hline
\end{tabular}
\caption{Datasets used for benchmarking}
\label{tab:ds3}
\end{table}

\subsection{Dataset Splitting}
For each dataset, we randomly split all the users into a training set (80\%) and a validation set (20\%), and then we  randomly divide each user data vector of the validation set into a feeding set (90\%) and a testing set (10\%). The training set simulates the Server's dataset, the feeding set simulates the rating data of the Client, and the testing set is used for accuracy evaluation. In the experiments, the Server trains recommendation models with its dataset. The Client sends its rating data vector to the Server, as a query, to get recommendations. For the models which have to be trained with the Client's data (the feeding set), we directly append the feeding set to the training set.  These models,  which require training from scratch with the Client's input, are identified in Section \ref{sec:modelclar}. For all the models, we repeat the accuracy evaluation experiments five times on each dataset. The root mean square error (RMSE) is adopted as the accuracy metric,
$$RMSE=\sqrt{\frac{\sum_{(u,i)\in \mathcal{D}}(\hat{r}_{ui}-r_{ui})^2}{|\mathcal{D}|}}  $$
where $\mathcal{D}$ is the testing set,  $|\mathcal{D}|$ is the number of ratings in the testing set. The lower the RMSE value, the higher the accuracy performance is.

$ $

\subsection{Remarks on Model Training}
\label{sec:modelclar}

\begin{table*}[h]
\centering
\begin{tabular}{|c|l|c|l|}
\hline
                 & \multicolumn{1}{c|}{NBM} & \multicolumn{1}{c|}{MF} & \multicolumn{1}{c|}{Neural Network Based} \\ \hline
 w/o Client & I-NBM\cite{desrosiers2011comprehensive}                    &        $\varnothing$                 & U-AutoRec\cite{sedhain2015autorec},\cite{strub2015collaborative,ouyang2014autoencoder,salakhutdinov2007restricted}                 \\ \hline
 w/ Client  & U-NBM\cite{desrosiers2011comprehensive}                    & BiasedMF\cite{koren2009matrix},\cite{koren2008factorization,mnih2008probabilistic}     &  I-AutoRec\cite{sedhain2015autorec},  \cite{he2017neural,zhang2017autosvd,kim2016convolutional,wang2015collaborative} \\ \hline
\end{tabular}
\caption{Remarks on model training }
\label{tab:models}
\end{table*}

\begin{table*}[t]
\centering
\begin{tabular}{|l|c|c|c|c|c|c|}
\hline
\multirow{2}{*}{} & \multicolumn{2}{c|}{netflix}             & \multicolumn{2}{c|}{ml1m}                & \multicolumn{2}{c|}{yahoo}               \\ \cline{2-7}
                  & RMSE                      &loss \%      & RMSE                      & loss \%      & RMSE                      &loss \%      \\ \hline
I-NBM             & 0.9115$\pm$0.007          & 9.4          & 0.8872$\pm$0.012          & 6.0          & 0.9899$\pm$0.017          & 0.2          \\ \hline
U-AutoRec         & 0.9762$\pm$0.012          & 17.1         & 0.9526$\pm$0.007          & 13.9         & 1.0621$\pm$0.014          & 7.5          \\ \hline
CryptoRec         & \textbf{0.8586$\pm$0.005} & \textbf{3.0} & \textbf{0.8781$\pm$0.007} & \textbf{4.9} & \textbf{0.9888$\pm$0.011} & \textbf{0.1} \\ \hline
\emph{I-AutoRec} & \emph{0.8334$\pm$0.006}& \emph{0} & \emph{0.8367$\pm$0.004} & \emph{0} & \emph{0.9880$\pm$0.015}& \emph{0}  \\ \hline
\end{tabular}
\caption{Accuracy comparison with pre-trained models. I-AutoRec is the accuracy benchmark.}
\label{tab:compwo}
\end{table*}

By investigating recommender systems which aim to provide accurate rating predictions, we informally classify these models into two categories as shown in Table \ref{tab:models}. the category ``w/o Client" contains the models which allow offering recommendations with a pre-trained model while the Client's private data is not in the training set; the category ``w/ Client" includes the models which have to be trained or re-trained with the Client's data. We refer interested readers to the two comprehensive reviews~\cite{su2009survey,zhang2017deep} for more details.

The models which fall into ``w/ Client" category often have one or both of the two following characteristics,  
\begin{itemize}
\item User and item features are jointly learned in the training phase, such as MF and its variants~\cite{koren2009matrix,koren2008factorization,mnih2008probabilistic}.
\item The input is an item rating vector ($\textbf{r}_i$), such as U-NBM~\cite{desrosiers2011comprehensive} and I-AutoRec~\cite{sedhain2015autorec}.
\end{itemize}

The models in the category of ``w/o Client"  often take as input a user preference vector (e.g., $\textbf{r}_u$). The personalized user features are automatically captured in the prediction phase, such as I-NBM~\cite{desrosiers2011comprehensive}, U-AutoRec~\cite{sedhain2015autorec}, and our proposed CryptoRec's model.

We select I-NBM, BiasedMF (the representatives of traditional recommender systems), and U-AutoRec, I-AutoRec (the representatives of neural network based recommender systems) as the comparison baselines.

\subsection{Accuracy Benchmark}
Without considering privacy,  the model I-AutoRec achieves state-of-the-art accuracy (RMSE) performance~\cite{sedhain2015autorec}.  As such, we adopt I-AutoRec as the accuracy benchmark model and train it from scratch in a standard machine learning setting. Table \ref{tab:bench} presents the accuracy performance of I-AutoRec on the selected datasets.

\begin{table}[h!]
\centering
\begin{tabular}{|l|l|l|l|}
\hline
\textbf{} & netflix          & ml1m             & yahoo            \\ \hline
I-AutoRec & 0.8334$\pm$0.006 & 0.8367$\pm$0.004 & 0.9880$\pm$0.015 \\ \hline
\end{tabular}
\caption{Accuracy benchmark (RMSE) on plaintext}
\label{tab:bench}
\end{table}

\subsection{Hyper-parameter Setting}
To train the CryptoRec's model, we perform a grid search for each hyper-parameter. In particular, for the learning rate $\eta$ we search in $\{0.0001, 0.0002, 0.0004 \}$; for the regular parameter $\lambda$ we search in $\{0.00001, 0.00002, 0.00004 \}$; for the dimension of the features $\{\textbf{A}, \textbf{Q}\}$, we search in $\{300, 400, 500, 600 \}$. As a result, we choose $\eta=0.0002$, $\lambda=0.00002$ and the dimension $d=500$. To train the baseline models, we also perform a grid search around the suggested settings given in their original papers, as the dataset splitting is not the same. By doing so, we have a fair comparison.

\section{Performance Evaluation}
\label{sec:pe}
In this section, we first evaluate and compare the accuracy and efficiency performance of using only a pre-trained model (Subsection~\ref{sec:wort}). Then we investigate the accuracy and efficiency performance of a re-training process (Subsection~\ref{sec:wrt}).

\subsection{Comparison with Pre-trained Models}
\label{sec:wort}
As described in Section \ref{sec:modelclar}, CryptoRec's model, I-NBM, and BiasedMF allow computing recommendation with a pre-trained model, and the Client's private data is not in their training set. In this section, we first verify and compare the accuracy performance by directly using the pre-trained models (without re-training the models with the Client's data). Then we analyze and compare the computational complexity of responding one prediction query in a private manner. Compared to the complexity of homomorphic operations, algebraic operations in the plaintext space are trivial. As such, for the computational complexity analysis, we only count in the operations over encrypted data,  i.e., operations between two ciphertexts (i.e., $\oplus, \otimes$) and multiplicative operations between a plaintext and a ciphertext (i.e., $\odot$).

\subsubsection{Accuracy Comparison}
Table \ref{tab:compwo} presents the accuracy performance of each model. Compared to the benchmark (Table \ref{tab:bench}), the accuracy of the three models is compromised to some extent (column loss\% ), and CryptoRec has the least loss. Specifically, CryptoRec loses 3.0\% accuracy on netflix, 4.9 \% on ml1m, and 0.1\% on yahoo. Clearly, CryptoRec is able to provide a promising accuracy guarantee to the Client by using only a pre-trained model.

\begin{table}[h!]
\centering
\begin{tabular}{|l|c|c|c|}
\hline
          & $\oplus$      & $\odot$        & Sigmoid \\ \hline
I-NBM     & $\mathcal{O}(m^2)$ & $\mathcal{O}(m^2)$    & $\varnothing$         \\ \hline
U-AutoRec & $\mathcal{O}(md)$     & $\mathcal{O}(md)$            & $\mathcal{O}(md)$       \\ \hline
CryptoRec      & $\mathcal{O}(md)$         & $\mathcal{O} (md)$           & $\varnothing$        \\ \hline
\end{tabular}
\caption{Computational complexity comparison of using pre-trained models.}
\label{tab:ccc}
\end{table}

\subsubsection{Computational Complexity Comparison}
To respond to a query from the Client, the Server has to predict the Client's preferences on all the items since it gets only an encrypted rating vector ($\llbracket \textbf{r}_u \rrbracket$). We analyze each model's computational complexity of answering one query, as shown Table \ref{tab:ccc}. Among the three models, I-NBM consumes more homomorphic additions ($\oplus$) and multiplications ($\odot$); U-AutoRec costs a similar number of $\oplus$ and $\otimes$ than CryptoRec's model,  but it introduces $\mathcal{O}(md)$ non-linear transformations (i.e., Sigmoid). Computing the Sigmoid function often relies on secure multiparty computation (SMC) schemes or polynomial-approximation~\cite{liu2017oblivious,gilad2016cryptonets}. The former requires the Server and Client to be online constantly and pay the price of extra communication overhead~\cite{bost2015machine,liu2017oblivious}; the latter leads to the use of a (somewhat) fully homomorphic encryption scheme since it introduces homomorphic multiplications between two ciphertexts ($\otimes$)~\cite{gilad2016cryptonets}. Apparently, CrytoRec yields the best efficiency performance.

\subsubsection{Evaluation of CryptoRec}
As shown in Table \ref{tab:ccc}, CryptoRec needs only homomorphic additions $\oplus$, and multiplications between ciphertexts and plaintexts $\odot$. As such, any additively homommorphic encryption can be employed to implement CryptoRec. In this paper, we adopt the Paillier cryptosystem~\cite{paillier1999public} implemented in the library python-paillier~\cite{lib:paillier}. In the implementation, we scale-up the parameter values to integers, the Client can obtain correct recommendations by simply sorting the prediction results. We let secret key size $l=2048$. In this setting, the message size of one encrypted rating $\llbracket r_{ui} \rrbracket$ is around 512 bytes, or 0.5 KB.

Following the pruning method proposed by~\cite{han2015deep}, we remove the model parameters of which the values are very close to zero (i.e., [$-5\times 10^{-4}, 5\times 10^{-4}$]), since these model parameters don't contribute to the final predictions. Then we quantify the values of the left model parameters to be 11 bits (2048 shared parameter values), of which we can reuse most of the related computations. It is worth mentioning that this approach does not compromise the accuracy, sometimes, it even leads to a slightly better accuracy performance. The same phenomenon has been also observed by some other works such as \cite{han2015deep,nowlan1992simplifying}.

\begin{figure}[h]
\centering
\includegraphics[scale=0.41]{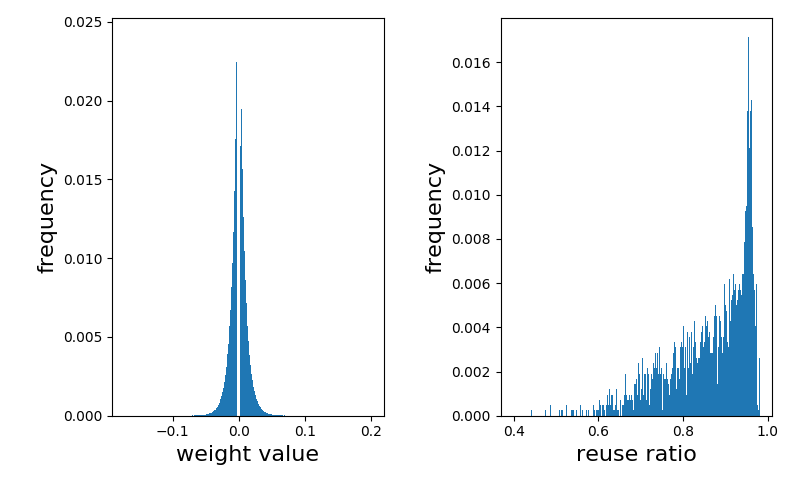}
\caption{CryptoRec's model trained on dataset ml1m: the distribution of parameter values after pruning (left panel); the distribution of the reuse ratio of each row of model parameters $\textbf{A}$ and $\textbf{Q}^T$ after the quantification (right panel). }
\label{fig:prune}
\end{figure}

The left panel of Figure \ref{fig:prune} intuitively describes the model parameter values distribution after the pruning, and the right panel of  Figure \ref{fig:prune} is the reuse ratio distribution of each row of model parameters $\textbf{A}$ and $\textbf{Q}^T$ after the quantification, the model here is trained on dataset ml1m. Table \ref{tab:prune} presents the \textbf{pruning ratio} and overall \textbf{reuse ratio} of the CryptoRec's model trained on each dataset, where we define the \textbf{pruning ratio} as $\frac{\#\ of\ pruned\ parameter}{ \#\ of\ all\ the\ parameter }$,  and compute the \textbf{reuse ratio} as $1-\frac{\#\ of\ unique\ parameter}{\#\ of\ all\ parameter}$. 

\begin{table}[h!]
\centering
\begin{tabular}{|l|l|l|l|}
\hline
\textbf{} & netflix          & ml1m             & yahoo            \\ \hline
pruning ratio &7.1\% &  9.2\% & 29.4\% \\ \hline
reuse ratio & 90.7\% & 90.5\% & 91.5\% \\ \hline
\end{tabular}
\caption{Parameter pruning ratio and computation reuse ratio of CryptoRec's models}
\label{tab:prune}
\end{table}

According to Table \ref{tab:prune}, we know that reusing computations on the shared parameter values is able to significantly reduce the computational complexity. For example, when computing $\llbracket r_{ui}\rrbracket \odot \textbf{A}_{j:} = \{\llbracket r_{ui}\rrbracket \odot \textbf{A}_{j1}, \llbracket r_{ui}\rrbracket \odot \textbf{A}_{j2},\cdots,\llbracket r_{ui}\rrbracket \odot \textbf{A}_{jd} \}$, we only need to compute  $\odot$ operations on each shared parameter value of $\textbf{A}_{j:}$ ($j$-th row of $\textbf{A}$), and then reuse the results at the other places of $\textbf{A}_{j:}$.

\begin{table}[h]
\centering
\begin{tabular}{|l|c|c|c|}
\hline
\textbf{}          & netflix & ml1m & yahoo \\ \hline
Communication (MB)  & 4.8    & 3.86 & 3.72  \\ \hline
Server time cost (s) & 14.2 & 10.9 & 7.3  \\ \hline
Client time cost (s) & 7.1 & 5.8 & 5.6 \\ \hline
\end{tabular}
\caption{The communication (MB) and time (s) cost of CryptoRec with a pre-trained model}
\label{tab:msgtime}
\end{table}

We summarize the communication and time cost of the Client and Server in Table \ref{tab:msgtime}. To elaborate the prediction process and the costs, we take the experiment on dataset ml1m as an example (We ignore the time cost of a public key pair generation, as it is trivial to the overall time cost),
\begin{itemize}
    \item Client:  Encrypting the rating vector $\llbracket \textbf{r}_u^{1\times 3952} \rrbracket$  takes 4.5 seconds. The message size of  $\llbracket \textbf{r}_u^{1\times 3952} \rrbracket$ is $0.5 \times 3952$ KB, or 1.93 MB.
    \item Server: Executing CryptoRec on $\llbracket \textbf{r}_u^{1\times 3952} \rrbracket$  takes 10.9 seconds. The message size of the output $\llbracket \hat{\textbf{r}}_u^{1\times 3952} \rrbracket$ is 1.93 MB.
    \item Client: Decrypting $\llbracket \hat{\textbf{r}}_u^{1\times 3952} \rrbracket$ takes 1.3 seconds.
\end{itemize}

We also implement the prediction process of I-NBM with the Paillier cryptosystem, where the item-item similarity matrix is pre-computed. Selecting the most similar $N$ items to a targeted item from a user's rating history is a typical approach used in I-NBM to compute recommendations. However, this approach introduces a number of extra non-linear operations (i.e., comparisons) which are not straightforwardly compatible with homomorphic encryption schemes. To address this issue, for each entry of the similarity matrix, we remove a certain number (e.g., 30\%) of elements which have the least values. The predictions computed on the sparsified similarity matrix are asymptotically close to the true predictions. In fact, using all the items for the prediction may lead to a significant accuracy loss. In our implementation, for one query, I-NBM requires 491 seconds, 335 seconds and 306 seconds on netflix, ml1m, yahoo datasets, respectively. We noted that Shmueli at al.~\cite{shmueli2017secure} used an additional mediator (i.e., a non-colluding global server) to achieve a more efficient solution. However, we focus on the 2PC protocol without using any third party, and in their setting, participants know which item to predict while in our case the Server doesn't know it. It is not necessary to include U-AutoRec in the comparison, because Sigmoid transformations it contains will result in a much worse efficiency performance.

CryptoRec's model allows providing accurate recommendations by a pre-trained model (the Client's data is not in the training set). So, the Server can provide recommendation services with a high throughput. In contrast, for the models which fall into category ``w/ Client", the time cost of the training process should be also counted, which leads to a notorious efficiency problem. For example, privately training matrix factorization on the dataset ml1m needs around 20 hours per iteration (more details are in Section \ref{sec:eoc}).

\begin{table*}[t]
\centering
\begin{tabular}{|l|c|c|c|c|c|c|c|c|c|}
\hline
\multirow{2}{*}{} & \multicolumn{3}{c|}{netflix}                           & \multicolumn{3}{c|}{ml1m}                              & \multicolumn{3}{c|}{yahoo}                              \\ \cline{2-10}
                  & RMSE                      & loss  \%      & iteration\# & RMSE                      & loss  \%     & iteration\# & RMSE                      &  loss  \%       & iteration\# \\ \hline
I-NBM             & 0.9061$\pm$0.005          & 8.7          & 1           & 0.8815$\pm$0.007          & 5.4          & 1           & 0.9853$\pm$0.014          & -0.3          & 1           \\ \hline
U-AutoRec         & 0.8849$\pm$0.007          & 6.2          & 35          & 0.8739$\pm$0.009          & 4.4          & 30          & 1.0583$\pm$0.016          & 7.1           & 26          \\ \hline
\emph{I-AutoRec}         & \emph{0.8334$\pm$0.006} & \emph{0}           & \emph{140}         & \emph{0.8367$\pm$0.004} & \emph{0}            & \emph{110}         & \emph{0.9880$\pm$0.015} & 0             & \emph{125}         \\ \hline
BiasedMF          & 0.8587$\pm$0.007          & 3.0          & 85          & 0.8628$\pm$0.009          & 3.1          & 80         & 0.9980$\pm$0.022          & 1.0           & 72          \\ \hline
CryptoRec         & \textbf{0.8391$\pm$0.006} & \textbf{0.7} & \textbf{12} & \textbf{0.8543$\pm$0.007} & \textbf{2.1} & \textbf{15} & \textbf{0.9821$\pm$0.013} & \textbf{-0.6} & \textbf{22} \\ \hline
\end{tabular}
\caption{Accuracy comparsion with a re-training step, I-AutoRec is the accuracy benchmark.}
\label{tab:accret}
\end{table*}

\begin{table*}[]
\centering
\begin{tabular}{|l|c|c|c|c|c|c|}
\hline
          & $\oplus$                      & $\otimes$                     & $\odot$                     & div                & Sigmoid                                & sqrt               \\ \hline
I-NBM     & $\mathcal{O}(m)$            & $\mathcal{O}(m^2)$            & $\mathcal{O}(m)$                 & $\mathcal{O}(m^2)$ & $\varnothing$       & $\mathcal{O}(m)$ \\ \hline
U-AutoRec & $\mathcal{O}(K(m+N_{\tau})d)$ & $\mathcal{O}(K(m+N_{\tau})d)$ & $\mathcal{O}(K\tau zd+md)$ & $\varnothing$      & $\mathcal{O}(Kmd)$             & $\varnothing$      \\ \hline
I-AutoRec & $\mathcal{O}(K(mn+N)d)$        & $\mathcal{O}(K(m+N)d)$        & $\mathcal{O}(Kmzd)$         & $\varnothing$      & $\mathcal{O}(Knd)$           & $\varnothing$      \\ \hline
BiasedMF  & $\mathcal{O}(Kmnd)$        & $\mathcal{O}(K(m+N)d)$        & $\mathcal{O}(md)$      & $\varnothing$      & $\varnothing$                 & $\varnothing$      \\ \hline
CryptoRec & $\mathcal{O}(K(m+N_{\tau})d)$ & $\mathcal{O}(K(m+N_{\tau})d)$ & $\mathcal{O}(K\tau zd+md)$ & $\varnothing$      & $\varnothing$                 & $\varnothing$      \\ \hline
\end{tabular}
\caption{Computational complexity comparison with a re-training step. $K$ is the number of training iterations. $N$ is the number of all observed ratings in the Server's dataset. $N_{\tau}$ is the number of observed ratings of the $\tau$ randomly selected users.  $z$ is the rating scale ($z=5$ in this paper). }
\label{tab:cccwr}
\end{table*}

\subsection{Comparison with a Re-training Step}
\label{sec:wrt}
In this section, we investigate the accuracy and efficiency performances of using the Client's  data to re-train a pre-learned CryptoRec's model. We first describe the details of re-training CryptoRec's model, then introduce a one-iteration training method for the sake of efficiency.

\subsubsection{ Re-training CryptoRec's Model}
\label{sec:traindetail}

\textbf{Avoiding Overfitting.} Using a single user's data to fine-tune a machine learning model learned from a large dataset may lead to an early overfitting. To address this issue, we re-train CryptoRec's model using the Client's data together with $\tau$ randomly selected users' data, where the $\tau$ users serve as a regularization term. We empirically set $\tau=10$, and $\tau \ll n$ ($n$ is the number of users in the Server's dataset).

\begin{figure}[h]
\centering
\includegraphics[scale=0.45]{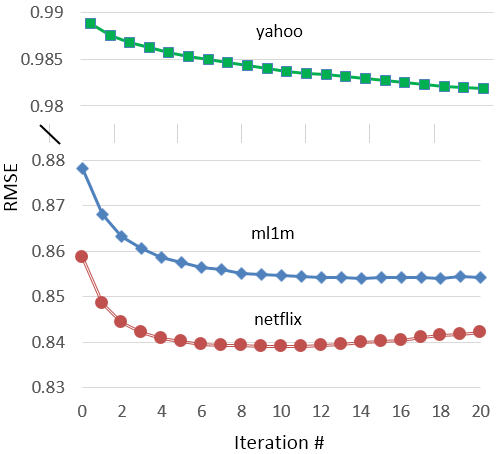}
\caption{Re-training CryptoRec with different iteration\# }
\label{fig:stop}
\end{figure}

\textbf{Stopping Criterion.} Identifying the stopping point of a training process over encrypted data is not as straightforward as doing that on clear data. This is because the Server gets only an encrypted  model, that the accuracy performance at each training iteration cannot be observed. To address this issue, the early-stopping strategy~\cite{prechelt1998early} can be a choice. Fortunately, we have also observed that, for the re-training process, the first several training iterations contribute most to the accuracy increase (RMSE decrease), as shown in Figure \ref{fig:stop} and Table \ref{tab:wovsw}. Specifically, the first training iteration leads to a big step towards the optimal accuracy performance. With 3 to 5 iterations, the accuracy performance can be asymptotically close to the best accuracy.  Therefore, we can conservatively re-train CryptoRec's model (e.g., 4 iterations) while still leading to a nearly consistent accuracy optimization.

\begin{table}[]
\centering
\begin{tabular}{|l|c|l|l|l|l|}
\hline
\multirow{2}{*}{} & \multicolumn{1}{c|}{\multirow{2}{*}{\begin{tabular}[c]{@{}c@{}}w/o retrain\\ (RMSE)\end{tabular}}} & \multicolumn{2}{c|}{retrain-full}                      & \multicolumn{2}{c|}{retrain-once}                      \\ \cline{3-6} 
                  & \multicolumn{1}{c|}{}                                                                              & \multicolumn{1}{c|}{RMSE} & \multicolumn{1}{c|}{inc\%} & \multicolumn{1}{c|}{RMSE} & \multicolumn{1}{c|}{inc\%} \\ \hline
netflix           & 0.8586                                                                                             & 0.8391                    & 2.3                        & 0.8485                    & 1.18                       \\ \hline
ml1m              & 0.8781                                                                                             & 0.8543                    & 2.7                        & 0.8680                    & 1.15                       \\ \hline
yahoo             & 0.9888                                                                                             & 0.9821                    & 0.7                        & 0.9874                    & 0.14                       \\ \hline
\end{tabular}
\caption{CryptoRec accuracy comparison: without retraining (w/o retrain) - retrain until convergence (retrain-full) -  retrain only once (retrain-once). ``inc\%" denotes  the percentage of the accuracy increase. }
\label{tab:wovsw}
\end{table}

\subsubsection{Accuracy comparison}
We summarize the accuracy performance of each model  in Table~\ref{tab:accret}, experiment results show that the accuracy performance of CryptoRec is competitive with the benchmark (as described in Table~\ref{tab:bench}) and consistently outperforms the other baseline models. Specifically, compared to the benchmark, CryptoRec loses 0.7\% accuracy on netflix; loses 2.1\% accuracy on ml1m; on yahoo, CryptoRec slightly outperforms the benchmark. Note that Table~\ref{tab:accret} presents the optimal accuracy performance of each model. In practice, the Server may achieve a suboptimal accuracy performance, due to the stopping point selection strategy or the constraint of computational resource. Roughly, predictions using only a pre-trained model reach the lower-bound of accuracy, the re-training process leads to a better accuracy performance. The optimal accuracy performance can guide users to perform the trade-off between efficiency and accuracy.

\subsubsection{Computational Complexity Comparison}
Table \ref{tab:cccwr} presents the computational complexity of each model. Among all the models, only MF and CryptoRec's model can be trained without using non-linear operations. However, MF has to be trained on the whole dataset which results in a serious efficiency issue. In contrast, re-training CryptoRec needs only the Client's data and the data of several randomly selected users (for regularization). We have noted that some researchers proposed incremental matrix factorization training methods such as \cite{vinagre2014fast,song2015incremental,huangincremental}. Unfortunately, these incremental training methods either require the Server to collect partial data of the Client \cite{vinagre2014fast,huangincremental} (we assume that the Server has no prior knowledge of the Client' rating data), or introduce extra non-linear operations~\cite{song2015incremental}. Therefore, we don't include these incremental matrix factorization training methods in the comparison. As presented in Table \ref{tab:cccwr}, CryptoRec shows a significant advantage in the efficiency performance.

\subsubsection{Evaluation of CryptoRec}
\label{sec:eoc}
Re-training CryptoRec's model needs a (somewhat) fully homomorphic encryption scheme (SWHE) since homomorphic addition $(\oplus)$ and multiplication ($\otimes$) are both required.  Some of the more significant advances in implementation improvements for SWHEs have come in the context of  the ring-learning-with-error (RLWE) based schemes, such as Fan-Vercauteren scheme~\cite{fan2012somewhat}. RLWE-based homomorphic encryption schemes map a plaintext message from ring $\mathfrak{R}_t^p:=\mathbb{Z}_t[x]/(x^p+1)$ to ring $\mathfrak{R}_t^p:=\mathbb{Z}_q[x]/(x^p+1)$ (ciphertext). The security level depends on the plaintext modulus $t$,  the coefficient modulus $q$, the degree $p$ of the polynomial modulus. In this paper, we adopt the Fan-Vercauteren scheme,  a real-number-supported version of which is implemented in the SEAL library~\cite{chen2017simple}. We set the polynomial degree $p=4096$,  the plaintext modulus $t=65537$, $q$ is automatically selected by the SEAL library given the degree $p$. To encode real numbers, we reserve 1024 coefficients of the polynomial for the integral part (low-degree terms) and expand the fractional part to 16 digits of precision  (high-degree terms). The circuit privacy is guaranteed by using \emph{ relinearization} operations~\cite[Section 8]{chen2017simple}. We refer interested readers to the paper \cite{chen2017simple} for more detail of the settings.

Compared to partial homomorphic encryption schemes such as the Paillier cryptosystem,  using an SWHE scheme results in a much larger ciphertext, which in turn leads to a higher computational complexity for a homomorphic operation. In this paper, the polynomial degree $p=4096$. Each coefficient of the polynomial costs 24 bytes (using SEAL)~\cite{gilad2016cryptonets}. So the size of a ciphertext is 4096*24 bytes or 96 KB. Taking the re-training process on dataset ml1m as an example, the item features $\llbracket \textbf{Q}^{3952\times 500} \rrbracket$ need 3952*500*96 KB or 181 GB RAM. Though it is not infeasible for a commercial server,  it is too expensive to respond to a single query while the accuracy improvement is limited.

\begin{algorithm}[h]
\caption{Re-train CryptoRec's model with one iteration}
\label{alg:retrain1}
%\begin{flushleft}
%\textbf{Input:   } 
%\end{flushleft}
%\algrule
\begin{algorithmic}[1]
\Procedure{Re-train}{$  \llbracket \textbf{r}_u \rrbracket,  \llbracket \boldsymbol{\phi}_u \rrbracket, \textbf{A}^{(0)}, \textbf{Q}^{(0)}, \lambda, \eta$} 
\State $\llbracket \textbf{y}_u  \rrbracket \gets \llbracket \textbf{r}_u \rrbracket \textbf{A}^{(0)}$
\State$\llbracket \textbf{e}_{u} \rrbracket \gets \llbracket \hat{\textbf{r}}_u \rrbracket \ominus  \llbracket \textbf{r}_u \rrbracket=\llbracket \textbf{y}_u  \rrbracket \textbf{Q}^{(0)} \ominus  \llbracket \textbf{r}_u \rrbracket$
\State $\llbracket \textbf{x}_u  \rrbracket \gets (\llbracket \textbf{e}_{u} \rrbracket \otimes   \llbracket \boldsymbol{\phi}_u \rrbracket)\textbf{Q}^{(0)}$
\For{\texttt{$j \gets \{ 1,2,\cdots,d \}$ }}
\State $\llbracket \Delta \textbf{A}_{:j} \rrbracket \gets  (\llbracket \textbf{x}_u  \rrbracket[j] \otimes  \llbracket \textbf{r}_{u}^T \rrbracket)  \oplus \lambda \cdot \textbf{A}^{(0)}_{:j} $ \Comment{gradient}
\State $\llbracket \textbf{A}_{:j}\rrbracket \gets  \textbf{A}^{(0)}_{:j} \ominus  (\eta \odot \llbracket \Delta \textbf{A}_{:j} \rrbracket$) \Comment{updates $\textbf{A}_{:j}$}
\State $\llbracket \textbf{p}_{u} \rrbracket[j] \gets \llbracket \textbf{r}_{u} \rrbracket \llbracket \textbf{A}_{:j} \rrbracket $ \Comment{computes user features}
\State release $ \llbracket \textbf{A}_{:j} \rrbracket$, $ \llbracket \textbf{x}_{u} \rrbracket[j]$
\EndFor
\State release  $\llbracket \textbf{r}_u \rrbracket$
\For{\texttt{$i \gets \{ 1,2,\cdots,m \}$ }}
\State $\llbracket  \Delta \textbf{q}_i \rrbracket  \gets \llbracket \boldsymbol{\phi}_{u} \rrbracket [i] \otimes ((\llbracket \textbf{e}_u \rrbracket[i] \otimes  \llbracket \textbf{y}_u  \rrbracket) \oplus \lambda \cdot \textbf{q}^{(0)}_i)$ 
\State  $ \llbracket \textbf{q}_{i}  \rrbracket \gets \textbf{q}_i^{(0)} \ominus  (\eta \odot \llbracket  \Delta \textbf{q}_i \rrbracket)$ \Comment{updates $\textbf{q}_{i}$}
\State $\llbracket \hat{\textbf{r}}_{u} \rrbracket[i] \gets  \llbracket \textbf{p}_u \rrbracket \llbracket \textbf{q}_i \rrbracket  $ \Comment{computes the prediction $\hat{r}_{ui}$}
\State release $ \llbracket \textbf{q}_{i} \rrbracket $, $\llbracket \textbf{e}_u \rrbracket[i]$, $\llbracket \boldsymbol{\phi}_u \rrbracket[i]$
\EndFor
\State \textbf{return} $ \llbracket \hat{\textbf{r}}_{u} \rrbracket $
\EndProcedure
\end{algorithmic}
\end{algorithm}

By exploiting the fact that the first re-training iteration contributes a big portion to the accuracy increase (Table \ref{tab:wovsw}), we introduce an efficient one-iteration re-training method, described in Algorithm \ref{alg:retrain1}\footnote {For simplicity, we omitted from Algorithm \ref{alg:retrain1} the bias terms and the $\tau$ number of randomly chosen users (Section \ref{sec:traindetail}). Note that the gradients computed from the data of the $\tau$ users and the pre-trained model parameters are plaintext. Therefore, all the operations related to the $\tau$ users are in plaintext, and have a trivial impact on the efficiency performance.}. The gradients of parameters are presented in Equation~(\ref{eq:grad}). The basic idea of this method is to timely release the model parameters which will not be used in the future (line 9, 10, 15, Algorithm \ref{alg:retrain1}). For example, we immediately release $\llbracket \textbf{A}_{:j} \rrbracket $ and $\llbracket \textbf{x}_{u} \rrbracket[j]$ after computing $\llbracket \textbf{p}_{u} \rrbracket[j]$  (line 9), where  $\textbf{A}_{:j} $ denotes $j$-th column of matrix $ \textbf{A}$. $ \textbf{Q}^{(0)}$ and $ \textbf{A}^{(0)}$ are pre-trained model parameters. $\llbracket \textbf{x} \rrbracket *  \llbracket \textbf{y} \rrbracket$ denotes $ \{ \llbracket x_i \rrbracket *  \llbracket y_i \rrbracket\}_m $ and $x *  \llbracket \textbf{y} \rrbracket$ denotes $ \{ x_i *  \llbracket y_i \rrbracket\}_m $ , where $*$ can be any operator such as $\oplus, \otimes$. $\ominus$ is homomorphic subtraction which can be implemented by $\oplus$. With Algorithm \ref{alg:retrain1}, we can complete the one-iteration training process (including computing the predictions) with less than 2 GB RAM.

\begin{table}[h]
\centering
\begin{tabular}{|l|c|c|c|}
\hline
\textbf{}          & netflix & ml1m & yahoo \\ \hline
Communication (GB)  & 1.31    & 1.08 & 1.04  \\ \hline
Server time cost (H) & 9.4 & 7.8 & 7.5  \\ \hline
Client time cost (s) & 14.3 & 11.8 &11.4  \\ \hline
\end{tabular}
\caption{The communication and time cost of CryptoRec with one-iteration re-training process}
\label{tab:msgtime2}
\end{table}

We summarize the communication and time costs of the Client and Server in Table \ref{tab:msgtime2}. We take the experiment on dataset ml1m as an example to introduce the cost on the two sides, respectively,
\begin{itemize}
   \item Client:  Encrypting the rating vector $\textbf{r}_u^{1\times 3952}$ and indication vector $\boldsymbol{\phi}_u^{1\times 3952}$   takes 9.6 seconds. The message size of  $\llbracket \textbf{r}_u^{1\times 3952} \rrbracket$ and $\llbracket \boldsymbol{\phi}_u^{1\times 3952} \rrbracket$  is $96 \times 3952 \times 2$ KB, or 741 MB.
    \item Server: Executing CryptoRec on $\llbracket \textbf{r}_u^{1\times 3952} \rrbracket$  takes 7.8 hours. The message size of the output is $\llbracket \hat{\textbf{r}}_u^{1\times 3952} \rrbracket$ is 370.5 MB.
    \item Client: Decrypting $\llbracket \hat{\textbf{r}}_u^{1\times 3952} \rrbracket$ takes 2.2 seconds.
\end{itemize}

%\begin{table*}[t]
%\centering
%\begin{tabular}{|l|l|c|c|c|c|}
%\hline
%\multirow{2}{*}{source} & \multirow{2}{*}{target} & \multicolumn{2}{c|}{no transfer learning} & \multicolumn{2}{c|}{transfer learning} \\ \cline{3-6}
%                        &                         & RMSE                   & iteration\#      & RMSE                 & iteration\#     \\ \hline
%ml1m*$_1$               & yahoo*                  & 0.8664$\pm$0.024       & 122              & 0.8312$\pm$0.016     & 32              \\ \hline
%yahoo*                  & ml1m*$_1$               & 0.8460$\pm$0.008       & 120              & 0.8463$\pm$0.011     & 55              \\ \hline
%neflix*                 & ml1m*$_2$               & 0.8527$\pm$0.003       & 118              & 0.8421$\pm$0.004     & 26              \\ \hline
%ml1m*$_2$               & neflix*                 & 0.8005$\pm$0.004       & 95               & 0.7981$\pm$0.005     & 28              \\ \hline
%\end{tabular}
%\caption{Transfer learning helps both accuracy (RMSE) and efficiency (\# of iterations)}
%\label{tab:tlae}
%\end{table*}

In contrast, the models which fall into ``w/ Client" category lead to a much higher time cost. For example, a recent work, GraphSC~\cite{nayak2015graphsc}, shows that a single iteration of training MF (the dimension of user/item features is 10), on the same dataset ml1m, took roughly 13 hours to run on 7 machines with 128 processors. In our setting, by making full use of the fact that the Server knows most of the users data, it still needs around 20 hours for any $i$-th iteration with 8 processors, where $i>1$. Worse, dozens of iterations are necessary for convergence \cite{koren2009matrix,nikolaenko2013privacy}.

\subsection{Discussion on Privacy and Scalability}
In this paper, we assume that the Server and Client should always agree upon a set of items, as it does not make sense to buy a service that the other party doesn't have, and vice versa. In fact, this assumption leads to a a trade-off between privacy and scalability. Informally, the more items that the two sides agreed on, the more privacy can be preserved. An online service provider (e.g., Youtube) may have millions of products, it is a notoriously challenging problem to provide recommendations from such a large corpus, even on clear data. A typical approach is to generate a small set of candidates, then compute recommendations from the candidates \cite{covington2016deep}. For our scenario, context information can be used to guide candidate generation, but still depending on whether such a context information is a privacy issue that the Client cares about (different users may have different concerns about privacy). The Server can train different recommendation models over datasets generated by different criteria such as children-friendly, place-of-origin, time-of-produce and so on. The Client can choose a criterion which doesn't violate her privacy concern, or choose multiple criteria at a time. How to design these criteria requires a further investigation of user preferences on privacy.

%% file: related.tex
\section{Related Work}
Canny et al.~\cite{canny2002collaborative} introduced a privacy-preserving solution for training collaborative filtering models (e.g., Singular Value Decomposition) in a peer-to-peer manner without assuming any trusted server. Nikolaenko et al.~\cite{nikolaenko2013privacy} proposed a garbled circuits~\cite{kolesnikov2008improved} based secure protocol to allow multiple users jointly train matrix factorization (MF), in which they assume two non-colluding servers. Shmueli et al.~\cite{shmueli2017secure} discussed that multi-party privately learn a neighborhood-based recommendation model by assuming a mediator that performs intermediate computations on encrypted data supplied by each party. Nayak at al.~\cite{nayak2015graphsc} brought parallelism to the secure implementation of oblivious version of graph-based algorithms (e.g., MF). Mohassel et al.~\cite{mohassel2017secureml}
further improved the efficiency of a secure framework with two non-colluding servers. The above solutions aim to privately train existing machine learning models. Different from these solutions, we aim to build a secure two-party computation protocol for Recommendation as a Service, without involving any third party (e.g., an additional non-colluding server). 

 Some recent works, such as ~\cite{gilad2016cryptonets,liu2017oblivious,rouhani2017deepsecure}, focused on neural network based  Machine Learning as a Service (MLaaS), the scenario of which is similar to ours. Their primary contribution is how to efficiently compute non-linear operations (e.g., comparison or Sigmoid function) on encrypted data. Gilad-Bachrach et al.~\cite{gilad2016cryptonets} substituted state-of-the-art activation functions such as ReLu ($relu(x)=max(0,x)$) with a simple square activation function ($f(x)=x^2$), this avoid the use of secure multiparty computation schemes. However, this approach often leads to a significant accuracy loss~\cite{liu2017oblivious,rouhani2017deepsecure}. To preserve the accuracy performance, Liu et al.~\cite{liu2017oblivious} and Rouhani et al.~\cite{rouhani2017deepsecure} proposed to evaluate neural networks with resort to secure multiparty computation schemes. Unfortunately, this approach requires the Client and Server to be online constantly. 

An orthogonal line of work focuses on constructing differentially private machine learning models, e.g., \cite{mcsherry2009differentially,liu2015fast,shokri2015privacy}. In their security models, a trusted server has full access to all the user data. It wishes to prevent adversaries from breaching the user privacy by exploiting the prediction results (i.e., inference attack).  In our security model, the Server learns nothing about client inputs; at the same time, the Client only learns what she can learn from the recommendation results. Our work and differential privacy~\cite{dwork2014algorithmic} can be complementary to each other.
%$$ $$

%% file: conclusion.tex
\section{Conclusions}
%With the rising interest in machine learning from both academia and industry, providing recommendations as a service has great potential impact in the years to come. However, confidentiality concerns have to be considered.

In this paper, we proposed CryptoRec, a new secure two-party computation protocol for Recommendation as a Service. CryptoRec encompasses a homomorphic encryption friendly recommender system. This model uses only addition and multiplication operations, that it is straightforwardly compatible with homomorphic encryption schemes. Moreover, it can produce recommendations by using a pre-trained model while the Client's data is not in the Server's training set. As demonstrated in the experiments, CryptoRec is able to provide recommendation services with a high throughput, while still standing up to state-of-the-art accuracy performance.